\documentclass[final,12pt,cleveref]{colt2024} 
\title[Sinkhorn Algorithm for Sequentially Composed Optimal Transports]{Sinkhorn Algorithm for Sequentially Composed Optimal Transports}
\usepackage{times}
\usepackage{mathtools}
\usepackage{tikz}
\usepackage{physics2}
\usephysicsmodule{ab}
\usepackage{todonotes}
\usepackage[T1]{fontenc}
\usepackage[utf8]{inputenc}
\usetikzlibrary{tikzmark,decorations.pathmorphing}

\makeatletter
\newcommand{\subalign}[1]{%
  \vcenter{%
    \Let@ \restore@math@cr \default@tag
    \baselineskip\fontdimen10 \scriptfont\tw@
    \advance\baselineskip\fontdimen12 \scriptfont\tw@
    \lineskip\thr@@\fontdimen8 \scriptfont\thr@@
    \lineskiplimit\lineskip
    \ialign{\hfil$\m@th\scriptstyle##$&$\m@th\scriptstyle{}##$\hfil\crcr
      #1\crcr
    }%
  }%
}
\makeatother

\coltauthor{%
 \Name{Kazuki Watanabe} \Email{kazukiwatanabe@nii.ac.jp}\\
 \addr National Institute of Informatics, Tokyo
 \AND
 \Name{Noboru Isobe} \Email{nobo0409@g.ecc.u-tokyo.ac.jp}\\
 \addr The University of Tokyo, Tokyo%
}

\newcommand{\ind}[2]{{#1}^{(#2)}}
\newcommand{\vect}[1]{\mathbf{#1}}

\newcommand{\real}{\mathbb{R}}
\newcommand{\Rnneg}{\mathbb{R}_{\geq 0}}
\newcommand{\Rpos}{\mathbb{R}_{+}}
\newcommand{\Rmat}[2]{\mathbb{R}^{#1 \times #2}}
\newcommand{\Rnnegmat}[2]{\Rnneg^{#1 \times #2}}
\newcommand{\Rposmat}[2]{\Rpos^{#1\times #2}}
\newcommand{\nset}[1]{[#1]}
\newcommand{\setdist}[1]{\Delta^{#1}}
\newcommand{\defeq}{\coloneq}
\newcommand{\ent}[1]{H\left(#1\right)}
\newcommand{\laglange}{L}
\newcommand{\opt}[1]{\widehat{#1}}
\newcommand{\nat}{\mathbb{N}}
\newcommand{\hmet}{d_{H}}
\newcommand{\diag}[1]{\mathrm{Diag}\left(#1\right)}
\DeclareMathOperator{\KL}{KL}
\DeclarePairedDelimiterX{\divbrace}[2]{(}{)}{%
    #1\;\delimsize\|\;#2%
}

\newcommand{\kld}{%
    \KL\divbrace%
}

\def\Set#1{\Setdef#1\Setdef}
\def\Setdef#1|#2\Setdef{\left\{#1\,\;\mathstrut\vrule\,\;#2\right\}}%
\renewcommand{\th}{%
    \ifmmode
        ^\mathrm{th}%
    \else%
        \textsuperscript{th}\xspace%
    \fi%
}
\makeatletter

\newenvironment{proofs}{%
  \renewcommand{\proofname}{Proof Sketch}\proof}{\endproof}

\begin{document}

\maketitle

\begin{abstract}%
  Sinkhorn algorithm is the de-facto standard approximation algorithm for optimal transport, which has been applied to a variety of applications, including image processing and natural language processing.
  In theory, the proof of its convergence follows from the convergence of the Sinkhorn--Knopp algorithm for the matrix scaling problem, and Altschuler et al. show that its worst-case time complexity is in near-linear time. 

  Very recently, sequentially composed optimal transports were proposed by Watanabe and Isobe as a hierarchical extension of optimal transports. 
  In this paper, we present an efficient approximation algorithm, namely Sinkhorn algorithm for sequentially composed optimal transports, for its entropic regularization. Furthermore, we present a theoretical analysis of the Sinkhorn algorithm, namely (i) its exponential convergence to the optimal solution with respect to the Hilbert pseudometric, and (ii) a worst-case complexity analysis for the case of one sequential composition. 
\end{abstract}


\section{Introduction}
Sinkhorn algorithm has emerged as a cornerstone method for approximating solutions to optimal transport (OT) problems.
Initially introduced by \cite{sinkhorn1967concerning} in the context of matrix scaling, this algorithm alternates matrix projections into affine constraints defined by two vectors.
\cite{Cuturi13} later adapted the Sinkhorn algorithm for OT approximation, highlighting its GPU-friendly parallel structure, which has significantly broadened its application in machine learning tasks, such as those demonstrated by \citep{DBLP:conf/icml/ArjovskyCB17,DBLP:journals/tmlr/0001FMHZRWB24} and others.
This adaptability has spurred extensive theoretical research to enhance the algorithm's applicability and performance.
Notably, \cite{AltschulerWR17} developed a more efficient variant by interpreting the Sinkhorn algorithm as a coordinate descent method in the dual space of two Lagrange multipliers.
The following two convergence results, which ensure the computational efficiency of the algorithm, are important theoretical results in machine learning theory.
The first is global convergence using the Hilbert projective metric by \citep{FRANKLIN1989717}, and the second is a complexity analysis of the number of elementary arithmetic operations required for a given error tolerance by \citep{AltschulerWR17}.
We note that the (Lagrange) duality observation mentioned above plays an important role in the proof of these convergence properties.

Another strength of the Sinkhorn algorithm lies in its applicability not only to (vanilla) OT but also to variants of OT with additional affine constraints.
For example, in martingale OT \citep{DBLP:journals/fs/BeiglbockHP13,10.1214/13-AAP925} and multi-marginal OT \citep{GangboAndrzej98,M2AN/Pass2015}, the Sinkhorn algorithm has been applied in \citep{chen2024convergencesinkhornsalgorithmentropic,DBLP:journals/siamsc/BenamouCCNP15,DBLP:journals/nm/BenamouCN19}.
Note that the same theoretical analysis as for OT cannot be applied directly to these variants.
This difficulty arises because the coordinate descent direction in the dual space becomes more complex as the constraints that need to be satisfied in the variant become more complex than in vanilla.

The purpose of this paper is to extend the Sinkhorn algorithm to a further variant of OT, namely, \emph{sequentially composed OT (SeqOT)} introduced in \cref{def:seq_composed_optimal_transport} later.
A basic example of SeqOT is given by the linear program of the following form:
\begin{equation}
    \begin{aligned}
    \min_{(\ind{\vect{P}}{1},\ind{\vect{P}}{2})\in\varPi(\vect{a},\vect{b})} &\left\langle\ind{\vect{C}}{1},\ind{\vect{P}}{1}\right\rangle+\left\langle\ind{\vect{C}}{2},\ind{\vect{P}}{2}\right\rangle,\\
    \varPi(\vect{a},\vect{b})\coloneqq&\Set{\ab(\ind{\vect{P}}{1},\ind{\vect{P}}{2})\in\Rnneg^{m\times m}\times\Rnneg^{m\times m}
    |
    {
    \begin{array}{c}
    \vect{a} = \ind{\vect{P}}{1}\vect{1}, \quad \vect{b} = {\ind{\vect{P}}{2}}^\top \vect{1}\\
    \text{and}\\
    {\ind{\vect{P}}{1}}^{\top}\vect{1} = \ind{\vect{P}}{2}\vect{1}
    \end{array}
    }
    },
  \end{aligned}
  \tag{BaseSeqOT}
  \label{eq:BaseSeqOT}
\end{equation}
where $\left\langle\bullet,\bullet\right\rangle$ stands for component-wise inner product, $\vect{1}$ is the $m$-dimensional all-ones vector, $\ind{\vect{C}}{1}$ and $\ind{\vect{C}}{2}$ are given cost matrices and $\vect{a}$ and $\vect{b}$ are given probability vectors.
SeqOT is a special case of the compositional OT proposed by \cite{WatanabeIsobe} for \emph{hierarchical planning}---it has been widely studied, mainly in the context of hierarchical reinforcement learning (see, e.g., \cite{deds/BartoM03,csur/PateriaSTQ21})---and SeqOT is a problem in which the sequence of transportation plans is subject to the constraint that composed plans have the same marginal distribution; see the last constraint of \eqref{eq:BaseSeqOT}.
This compositional constraint is a major obstacle when applying Sinkhorn to SeqOT because it is not obvious how to make the \emph{sequential} adjacency structure compatible with the \emph{parallel} structure of the conventional Sinkhorn algorithm.
\paragraph{Contributions.} This paper makes three contributions.
This paragraph presents these contributions only in the case of \eqref{eq:BaseSeqOT} for simplicity.
Note that in subsequent sections, we will also consider the more general SeqOT.
\begin{enumerate}
    \item We propose a novel Sinkhorn-type algorithm for SeqOT in \cref{def:Sinkhorn}.
    If we write down the proposed algorithm in the case of \eqref{eq:BaseSeqOT}, we get the following:
    \begin{achievement}[An example of \cref{def:Sinkhorn}]\label{achivement:Sinkhorn}
        For \eqref{eq:BaseSeqOT}, given vectors $(\vect{u}^0,\vect{v}^0,\vect{w}^0)\in\Rpos^{m}\times\Rpos^{m}\times\Rpos^{m}$ are updated recursively as follows:
        \begin{align*}
        \vect{w}^{n+1} \leftarrow \sqrt{ 
        \frac{(\ind{\vect{K}}{1})^{\top}
        \vect{u}^{n}}{\ind{\vect{K}}{2}
        \vect{v}^{n} }
        },
        &&\vect{u}^{n+1} \leftarrow \frac{\vect{a}}{ \ind{\vect{K}}{1}
        (1/\vect{w}^{n+1})  },
        &&\vect{v}^{n+1} \leftarrow \frac{\vect{b}}{(\ind{\vect{K}}{2})^{\top}
        \vect{w}^{n+1} },
    \end{align*}
    for $n=0,1,\dots.$ 
    Here, $\ind{\vect{K}}{1}$ and $\ind{\vect{K}}{2}$ are the Gibbs kernel induced from $\ind{\vect{C}}{1}$ and $\ind{\vect{C}}{2}$, respectively, and principal square roots and division work with vectors component-wisely.
    \end{achievement}
    Unlike Sinkhorn for OT, this algorithm alternately projects onto the adjacency constraint set and the marginal constraint set.
    This alternation matches the sequential structure with the parallel structure of matrix scaling and enables the extension of the Sinkhorn algorithm to SeqOT.
    
    \item We establish a global convergence result for the Hilbert metric for the proposed algorithm in \cref{thm:global_convergence}.
    For the case of \eqref{eq:BaseSeqOT}, we have the following:
    \begin{achievement}[Informal version of \cref{thm:global_convergence}]
        The sequence $((\vect{u}^n,\vect{v}^n,\vect{w}^n))_{n\in\nat}$ in \cref{achivement:Sinkhorn} converges exponentially to the optimal vectors $(\widehat{\vect{u}},\widehat{\vect{v}},\widehat{\vect{w}})$ in the sense of the Hilbert metric.
    \end{achievement}
    In this result, the technical contribution is that the proposed algorithm is designed in such a way that the elegant properties of the Hilbert metric are utilized. Furthermore, for \eqref{eq:BaseSeqOT}, the proposed algorithm achieves a faster convergence rate than the general SeqOT in \cref{thm:global_convergence_2}.
    \item  Finally, \cref{thm:main_theorem} below provides the computational complexity of the proposed method in the sense of \cite{AltschulerWR17} for \eqref{eq:BaseSeqOT}.
    \begin{achievement}[Simplified version of \cref{thm:main_theorem}]\label{achievement:complexity}
    Given a tolerance error $\delta>0$, our algorithm returns a pair $(\widetilde{\ind{\vect{P}}{1}},\widetilde{\ind{\vect{P}}{2}})\in\varPi(\vect{a},\vect{b})$ satisfying \[\Big\langle\ind{\vect{C}}{1}, \widetilde{\ind{\vect{P}}{1}} \Big\rangle + \Big\langle\ind{\vect{C}}{2}, \widetilde{\ind{\vect{P}}{2}} \Big\rangle \leq \min\limits_{(\ind{\vect{P}}{1},\ind{\vect{P}}{2})\in\varPi(\vect{a},\vect{b})}\ab\{\left\langle\ind{\vect{C}}{1},\ind{\vect{P}}{1}\right\rangle+\left\langle\ind{\vect{C}}{2},\ind{\vect{P}}{2}\right\rangle\} + \delta,\]
    within a complexity of $
        O(m^2(\log m) (\|\ind{\vect{C}}{1}\|_{\infty}+ \|\ind{\vect{C}}{2}\|_{\infty})^3 \delta^{-3})
    $.

    \end{achievement}
    In contrast, the existing methods in \citep{WatanabeIsobe} require at least $O(m^3)$.
    This discrepancy arises because our algorithm iteratively utilizes matrix-vector multiplication as demonstrated in \cref{achivement:Sinkhorn}, whereas the latter employs a matrix-matrix product.
    Technically, the manner in which we define the stopping criterion for the iteration is crucial. In vanilla OT, the distance from the target marginal served as the criterion. However, for SeqOTs, it is essential to establish a completely different criterion.
\end{enumerate}

\paragraph{Organizations of this paper.}
In this paper, we propose a Sinkhorn-type algorithm for general sequential composed optimal transport problems (SeqOTs), including the specific case described in \eqref{eq:BaseSeqOT}.
The formulation of SeqOTs and the proposed algorithm are detailed in \cref{sec:formulation}.
We establish the global convergence of our algorithm within the Hilbert metric framework in \cref{sec:global_conv}.
Subsequently, in \cref{sec:time_complexity}, we analyze the computational complexity of general SeqOTs with a single sequential composition, such as the BaseSeqOT.
Related work is discussed in \cref{sec:related}, and our conclusions are presented in \cref{sec:concllusion}.

\paragraph{Notation.}
We write $\Rnneg$ for the set of nonnegative real numbers and write $\Rpos$ for the set of positive real numbers.
We write $\setdist{m}$ for the probability simplex $\setdist{m}\defeq \{\vect{a} = (a_i)^m_{i=1}\in \Rnneg^m\mid \sum^m_{i=1} a_i = 1\}$.
Given two matrices $\vect{X}, \vect{Y}\in \Rmat{m_1}{m_2}$, we write $\langle\vect{X}, \vect{Y}\rangle\coloneq\sum_{i,j}X_{ij}\cdot Y_{ij}$ for the inner product. 
Given two vectors $\vect{x}, \vect{y}\in \real^m$, we write $\vect{x} \odot \vect{y}\in \real^m$ for the Hadamard product $\big(x_i\cdot y_i\big)_{i\in \nset{m}}$. 
Given a matrix $\vect{X}\in \Rmat{m}{n}$, the norm $\|X\|_1$ is defined by $\sum_{i\in \nset{m}, j\in \nset{n}} |X_{ij}|$. 
Given natural numbers $M$ and $N$, the set $\nset{M}$ is given by $\nset{M}\defeq \{1, \dots, M\}$, and the set $\nset{M, N}$ is given by $\nset{M, N}\defeq \{M, \dots, N\}$.

\section{Problem Formulation and the Sinkhorn Algorithm}\label{sec:formulation}

We fix the natural number $M\geq 2$ in this section. We first define our target problem, namely the \emph{sequentially composed optimal transport problem}. 
\begin{definition}[\cite{WatanabeIsobe}]
  \label{def:seq_composed_optimal_transport}
  Let $\big(\ind{\vect{C}}{i}\big)_{i\in \nset{M}} \in \prod_{i\in \nset{M}}  \Rnnegmat{m_i}{m_{i+1}}$ be an $M$-tuple of \emph{cost matrices}, and $\vect{a}\in \setdist{m_1}$ and $\vect{b}\in \setdist{m_{M+1}}$ be distributions. 
  The \emph{sequentially composed optimal transport problem} of the tuple $\big(\big(\ind{\vect{C}}{i}\big)_{i\in \nset{M}}, \vect{a}, \vect{b}\big)$ is defined by 
  \begin{align*}
    &\min_{\big(\ind{\vect{P}}{i}\big)\in \prod_{i\in \nset{M}}  \Rnnegmat{m_i}{m_{i+1}}} \sum_{i=1}^M \big\langle \ind{\vect{C}}{i}, \ind{\vect{P}}{i} \big\rangle = \sum_{i=1}^M \sum_{j=1}^{m_i} \sum_{k=1}^{m_{i+1}} \ind{C}{i}_{jk} \ind{P}{i}_{jk}\\
    \textup{subject to} \quad &\vect{a} = \ind{\vect{P}}{1}\vect{1}_{m_2}, \quad \vect{b} = {\ind{\vect{P}}{M}}^\top \vect{1}_{m_M},\quad\text{and}\quad {\ind{\vect{P}}{i}}^{\top}\vect{1}_{m_{i}} = \ind{\vect{P}}{i+1}\vect{1}_{m_{i+2}},
  \end{align*}
  for all $i\in \nset{M-1}$. 
\end{definition}
A feasible solution, which we call \emph{transportation plans}, of this problem is a tuple $\big(\ind{\vect{P}}{i}\big)_{i\in \nset{M}}$ of matrices.  
We provide an intuition of the three consistency conditions in the following: (i) the first transportation plan $\ind{\vect{P}}{1}$ is consistent with the distribution $\vect{a}$; (ii) the last transportation plan $\ind{\vect{P}}{M}$ is consistent with the distribution $\vect{b}$; and (iii) in each boundary between two transportation plans $\ind{\vect{P}}{i}$ and $\ind{\vect{P}}{i+1}$, they should be consistent, that is, ${\ind{\vect{P}}{i}}^{\top}\vect{1}_{m_{i}} = \ind{\vect{P}}{i+1}\vect{1}_{m_{i+2}}$. We note that~\cite{WatanabeIsobe} also introduced the parallel composition, while we focus on the sequential composition in this paper.  

Following the standard approximation approach for optimal transports with the entropic regularization~\citep{Cuturi13}, we introduce an entropic regularization for sequentially composed optimal transports.   

\begin{definition}
  \label{def:regularized_seq_composed_optimal_transport}
  Assume the setting of Def.~\ref{def:seq_composed_optimal_transport}, and let $\epsilon > 0$. 
  The \emph{regularized sequentially composed optimal transport problem} of the tuple $\big(\big(\ind{\vect{C}}{i}\big)_{i\in \nset{M}}, \vect{a}, \vect{b}\big)$ with $\epsilon$ is defined by 
  \begin{align*}
    &\min_{\big(\ind{\vect{P}}{i}\big)_{i\in \nset{M}}} \sum_{i=1}^M \big\langle \ind{\vect{C}}{i}, \ind{\vect{P}}{i} \big\rangle -\epsilon\ent{\ind{\vect{P}}{i}}  \\
    \textup{subject to} \quad &\vect{a} = \ind{\vect{P}}{1}\vect{1}_{m_2}, \quad \vect{b} = {\ind{\vect{P}}{M}}^\top \vect{1}_{m_M},\quad\text{and}\quad {\ind{\vect{P}}{i}}^{\top}\vect{1}_{m_{i}} = \ind{\vect{P}}{i+1}\vect{1}_{m_{i+2}} \quad\text{for all $i\in \nset{M-1}$, }
  \end{align*}
  where the \emph{entropy} $\ent{\vect{P}}$ of a matrix $\vect{P}\in \Rnnegmat{m}{m'}$ is defined by \[\ent{\vect{P}}\defeq -\sum_{j=1}^m\sum_{k=1}^{m'} P_{jk}\big((\log P_{jk}) - 1\big).\]
  Following the standard convention, we set $0\log 0 = 0$.
\end{definition}

By the standard argument of convex optimization (e.g.~\cite{boyd2004convex}), there is the Lagrange dual problem, and the strong duality holds. 
\begin{proposition}
  \label{prop:dual_regularized_seq_composed_optimal_transport}
  Assume the setting of Def.~\ref{def:regularized_seq_composed_optimal_transport}. 
  Its Lagrange dual problem is given by 
  \begin{align*}
    \max_{\big(\ind{\vect{f}}{i} \big)\in \prod_{i\in \nset{M+1}}\Rnneg^{m_i}}& \langle \ind{\vect{f}}{1}, \vect{a}\rangle + \langle \ind{\vect{f}}{M+1}, \vect{b}\rangle - \epsilon \Big(\sum^{m_{M}}_{j=1}\sum^{m_{M+1}}_{k=1} \exp\big((\ind{f}{M}_j + \ind{f}{M+1}_k-\ind{C}{M}_{jk})/\epsilon\big) \\
    & + \sum_{i=1}^{M-1} \sum^{m_i}_{j=1}\sum^{m_{i+1}}_{k=1} \exp\big( (\ind{f}{i}_j - \ind{f}{i+1}_k-\ind{C}{i}_{jl})/\epsilon\big) \Big).
  \end{align*}
  We also write $\max_{(\ind{\vect{f}}{i})_{i\in \nset{M+1}}}\laglange\big(\ind{\vect{f}}{i}\big)$ for the dual problem with the Lagrangian $\laglange\big(\ind{\vect{f}}{i}\big)$.
\end{proposition}
The proof is routine since it satisfies Slater's condition. See \S\ref{app:proof_prop_dual_regularized_seq_composed_optimal_transport} for the proof.

Based on the dual problem, we derive our novel Sinkhorn algorithm in a similar manner to coordinate descent algorithms: the idea is that we update the variables by solving the equation $\partial \laglange/\partial \ind{f}{i}_{j} = 0$ for each $i$ and $j$. 
\begin{definition}
  Given initial vectors  $(\ind{\vect{f}}{0, i})_{i\in \nset{M+1}}$, 
  the \emph{Sinkhorn iteration} $\big( (\ind{\vect{f}}{n, i})_{i\in \nset{M+1}} \big)_{n\in \nat}$ in log space is defined as follows: 
  \begin{align*}
    \ind{f}{n+1, i}_j &\defeq \frac{\epsilon}{2}\log\Bigg( \frac{\sum_{k\in \nset{m_{i-1}}} \exp\big((\ind{f}{n, i-1}_k - \ind{C}{i-1}_{kj})/\epsilon\big) }{\sum_{k\in \nset{m_{i+1}}} \exp\big( (-\ind{f}{n, i+1}_k - \ind{C}{i}_{jk})/\epsilon\big)}\Bigg) &&\text{ $\forall i\in \nset{2, M-1}$ and $\forall j\in \nset{m_i}$,}\\
    \ind{f}{n+1, M}_j &\defeq \frac{\epsilon}{2}\log\Bigg( \frac{\sum_{k\in \nset{m_{M-1}}} \exp\big((\ind{f}{n, M-1}_k - \ind{C}{M-1}_{kj})/\epsilon\big) }{\sum_{k\in \nset{m_{M+1}}} \exp\big( (\ind{f}{n, M+1}_k - \ind{C}{M}_{jk})/\epsilon\big)}\Bigg) &&\text{ $\forall j\in \nset{m_M}$,}\\
    \ind{f}{n+1, 1}_j &\defeq \epsilon \log (a_j) - \epsilon \log\Big(\sum_{k\in \nset{m_2}} \exp\big( (-\ind{f}{n+1, 2}_k - \ind{C}{1}_{jk})/\epsilon\big)\Big) &&\text{ $\forall j\in \nset{m_{2}}$,}\\
    \ind{f}{n+1, M+1}_j &\defeq \epsilon \log (b_j) - \epsilon  \log\Big(\sum_{k\in \nset{m_M}} \exp\big( (\ind{f}{n+1, M}_k - \ind{C}{M}_{jk})/\epsilon\big)\Big) &&\text{ $\forall j\in \nset{m_{M+1}}$.}
  \end{align*}
\end{definition}
Unlike the standard Sinkhorn algorithm for OTs, the iteration is mainly divided into the two parts: (i) we first update the variables $\ind{\vect{f}}{n+1, i}$ of boundaries (for $i\in \nset{2, M}$) with the previous values $(\ind{\vect{f}}{n, j})_{j\in \nset{M+1}}$ (in the first and second lines); and (ii) the variables $\ind{\vect{f}}{n+1, 1}$ and $\ind{\vect{f}}{n+1, M+1}$ on the edges are updated by the current values $\ind{\vect{f}}{n+1, 2}$ and $\ind{\vect{f}}{n+1, M}$ (in the third and fourth lines).

\begin{remark}
We carefully choose this ordering of the variable updates, that is, we first update the variables on the boundaries, and then update the variable on the edges. In fact, most proofs in this paper explicitly depend on this ordering, and it is unclear whether we can get similar results by changing the ordering.  
\end{remark}

By definining new vectors $\ind{\vect{u}}{n, i} \defeq \exp\big(\ind{\vect{f}}{n, i}/\epsilon\big)$ and \emph{Gibbs kernels} $\ind{\vect{K}}{i}\in \Rposmat{m_{i}}{m_{i+1}}$ as $\ind{K}{i}_{jk}\defeq  \exp(-\ind{C}{i}_{jk}/\epsilon)$ for any $j\in \nset{m_i}$ and $k\in \nset{m_{i+1}}$, we obtain a novel Sinkhorn iteration. 

\begin{definition}\label{def:Sinkhorn}
  Given initial vectors  $(\ind{\vect{u}}{0, i})\in \prod_{i\in \nset{M+1}} \Rpos^{m_{i}}$, 
  the \emph{Sinkhorn iteration} $\big( (\ind{\vect{u}}{n, i} )_{i\in \nset{M+1}} \big)_{n\in \nat}$ is defined as follows: 
  \begin{align*}
    \ind{\vect{u}}{n+1, i} &\defeq \Bigg(\frac{(\ind{\vect{K}}{i-1})^{\top} \ind{\vect{u}}{n,i-1}}{\ind{\vect{K}}{i}\frac{\vect{1}_{m_{i+1}}}{\ind{\vect{u}}{n,i+1}}}\Bigg)^{1/2}, \quad &&\ind{\vect{u}}{n+1, M} \defeq \Bigg(\frac{(\ind{\vect{K}}{M-1})^{\top} \ind{\vect{u}}{n,M-1}}{\ind{\vect{K}}{M}\ind{\vect{u}}{n,M+1}}\Bigg)^{1/2},  \\
    \ind{\vect{u}}{n+1, 1} &\defeq \frac{\vect{a}}{\ind{\vect{K}}{1}\frac{\vect{1}_{m_{2}}}{\ind{\vect{u}}{n+1,2}}}, \quad  &&\ind{\vect{u}}{n+1, M+1} \defeq \frac{\vect{b}}{(\ind{\vect{K}}{M})^{\top}\ind{\vect{u}}{n+1,M}},
  \end{align*}
  where $i\in \nset{2, M-1}$. 
\end{definition}

We conclude this section by introducing the matrices 
$\big((\ind{\vect{P}}{n, i})_{i\in \nset{M}}\big)_{n\in \nat}$ that are induced by the Sinkhorn iteration, which are, roughly speaking, approximations of optimal transportation plans. 
\begin{definition}
  Given a Sinkhorn iteration $\big( (\ind{\vect{u}}{n, i})_{i\in \nset{M+1}} \big)_{n\in \nat}$, the \emph{matrices $\big((\ind{\vect{P}}{n, i})_{i\in \nset{M}}\big)_{n\in \nat}$ induced by the iteration} are defined as follows: 
  \begin{align*}
    \ind{\vect{P}}{n, i} \defeq \diag{\ind{\vect{u}}{n, i}}\ind{\vect{K}}{i}\diag{\frac{\vect{1}_{m_{i+1}}}{\ind{\vect{u}}{n, i+1}}}, \quad  \ind{\vect{P}}{n, M} \defeq \diag{\ind{\vect{u}}{n, M}}\ind{\vect{K}}{M}\diag{\ind{\vect{u}}{n, M+1}}, 
  \end{align*}
  for any $n\in \nat$ and $i\in \nset{M}$.
\end{definition}

\section{Convergence}\label{sec:global_conv}

In this section, we show that the Sinkhorn iteration converges exponentially to the solution with respect to the Hilbert (pseudo-)metric from any initial vectors. We present two results: one is for the case of $M = 2$, and the other is for general cases $(M \geq 2)$. 
\begin{definition}[e.g.~\cite{busemann2012projective}]
  \label{def:hmet}
  The \emph{Hilbert (pseudo-)metric} $\hmet$ over $\Rpos^m$ is  defined by 
  \begin{align*}
    \hmet(\vect{u}, 
    \vect{v}) \defeq \log \Big(\max_{i\in \nset{m}}\Big(\frac{u_i}{v_i}\Big) \cdot \max_{j\in \nset{m}}\Big(\frac{v_j}{u_j}\Big) \Big).
  \end{align*} 
\end{definition}
We note that given $\vect{u}, \vect{v}\in \Rpos^m$. the equation $ \hmet(\vect{u}, 
    \vect{v}) = 0$ holds iff there exists a (necessarily unique) constant $c\in \Rpos$ such that $\vect{u} = c\vect{v}$.

A key lemma to prove the convergences is the Birkhoff's inequality, which is also used to prove the convergence of the Sinkhorn-Knopp algorithm~\citep{sinkhorn1967concerning,Knight08}.

\begin{lemma}[\cite{Birkhoff,Samelson,CAVAZOSCADENA2003291}]
  \label{lem:Birkhoff}
  Let $\vect{v}, \vect{v}'\in \Rpos^{m}$, and $\vect{A}\in \Rposmat{n}{m}$. 
  We set the value $\gamma(\vect{A})\in \Rpos$ and $\lambda(\vect{A})\in \real$ as follows: 
  \begin{align*}
    \gamma(\vect{A})\defeq \max_{i, j\in \nset{n}, k, l\in \nset{m}} \frac{A_{ik} A_{jl}}{A_{jk}A_{il}}, \quad \lambda(\vect{A})\defeq \frac{\big(\gamma(A)\big)^{1/2} - 1}{\big(\gamma(A)\big)^{1/2} + 1}. 
  \end{align*}
  Then, the value $\lambda(\vect{A})$ satisfies $0\leq \lambda(\vect{A}) < 1$ and the following inequality holds: 
  \begin{align*}
    \hmet(\vect{A}\vect{v}, \vect{A}\vect{v}') &\leq \lambda(\vect{A})\cdot \hmet(\vect{v}, \vect{v}').  &&\blacksquare
  \end{align*}  
\end{lemma}
 We also exploit the following characterization of optimal transportation plans. A variant for the standard OT is also used in~\cite{Cuturi13}.
\begin{lemma}
  \label{lem:ch_optimal_matrix}
  The following statements are equivalent: 
  \begin{itemize}
    \item The transportation plans $(\ind{\vect{P}}{i})_{i\in \nset{M}}$ are optimal for the primal problem defined in Def.~\ref{def:regularized_seq_composed_optimal_transport}.
    \item There are vectors $(\ind{\vect{x}}{i})_{i\in \nset{M}}$ and transportation plans $(\ind{\vect{P}}{i})_{i\in \nset{M}}$ such that 
    \begin{align*}
      \ind{\vect{P}}{i} &= \diag{\ind{\vect{x}}{i}}\ind{\vect{K}}{i}\diag{\frac{\vect{1}_{m_{i+1}}}{\ind{\vect{x}}{i+1}}} &\text{ for any $i\in \nset{M-1}$,}\\
      \ind{\vect{P}}{M} &= \diag{\ind{\vect{x}}{M}}\ind{\vect{K}}{M}\diag{\ind{\vect{x}}{M+1}}. &&\blacksquare
    \end{align*}
  \end{itemize}
\end{lemma}
The proof follows from the Karush--Kuhn--Tucker (KKT) theorem; see \S\ref{sec:proofChOTmat} for the proof.

\begin{corollary}
  \label{cor:ch_optimal_matrix}
  Assume that the vectors $(\ind{\opt{\vect{f}}}{i})_{i\in \nset{M+1}}$ are optimal for the Lagrange dual problem given in Prop.~\ref{prop:dual_regularized_seq_composed_optimal_transport}.  
  Then, we  can write the optimal transportation plans $(\ind{\opt{\vect{P}}}{i})_{i\in \nset{M}}$ that are induced by $(\ind{\opt{\vect{f}}}{i})_{i\in \nset{M+1}}$  as follows: 
  \begin{align*}
    \ind{\opt{\vect{P}}}{i} &= \diag{\ind{\opt{\vect{u}}}{i}}\ind{\vect{K}}{i}\diag{\frac{\vect{1}_{m_{i+1}}}{\ind{\opt{\vect{u}}}{i+1}}} &\text{ for any $i\in \nset{M-1}$,}\\
    \ind{\opt{\vect{P}}}{M} &= \diag{\ind{\opt{\vect{u}}}{M}}\ind{\vect{K}}{M}\diag{\ind{\opt{\vect{u}}}{M+1}},
  \end{align*}
  where $\ind{\opt{\vect{u}}}{i} \defeq \exp(\ind{\opt{\vect{f}}}{i}/\epsilon)$ for any $i\in \nset{M+1}$. 
\end{corollary}
\begin{proof}
It is an immediate consequence of the proof of Lem.~\ref{lem:ch_optimal_matrix}. 
\end{proof}

\subsection{The Case of $M = 2$}
For the case of $M = 2$, we show the following convergence result. 
\begin{theorem}[convergence]
  \label{thm:global_convergence_2}
  Let  $\big( (\ind{\vect{u}}{n, i})_{i\in \nset{3}} \big)_{n\in \nat}$ be a Sinkhorn iteration, and assume that given vectors $(\ind{\opt{\vect{u}}}{i})_{i\in \nset{3}}$ 
  become an optimal solution $\big(\ind{\opt{\vect{f}}}{i}\big)_{i\in \nset{3}}$ of the Lagrange dual problem defined in Prop.~\ref{prop:dual_regularized_seq_composed_optimal_transport} by translating $\ind{\opt{\vect{f}}}{i}\defeq \epsilon \log(\ind{\opt{\vect{u}}}{i})$ for each $i\in \nset{3}$. 
  For each $n\in \nat\backslash \{0\}$, the following inequalities hold: 
  \begin{align*}
    \hmet\Big(\ind{\vect{u}}{n, 1}, \ind{\opt{\vect{u}}}{1}\Big) \leq D^{2n}\cdot d, \quad \hmet\Big(\ind{\vect{u}}{n, 3}, \ind{\opt{\vect{u}}}{3}\Big) \leq D^{2n}\cdot d, \quad \hmet\Big(\ind{\vect{u}}{n, 2}, \ind{\opt{\vect{u}}}{2}\Big) \leq D^{2n-1}\cdot d,
  \end{align*}
  where $D \defeq \max \big(\lambda(\ind{\vect{K}}{1}), \lambda(\ind{\vect{K}}{2})\big)$ and $d \defeq \max\Big(\hmet\big(\ind{\vect{u}}{0, 1}, \ind{\opt{\vect{u}}}{1}\big), \hmet\big(\ind{\vect{u}}{0, 3}, \ind{\opt{\vect{u}}}{3}\big)\Big)$. 
  We note that $D < 1$ by definition (see Lem.~\ref{lem:Birkhoff}). 
\end{theorem}
\begin{proofs}
By Cor.~\ref{cor:ch_optimal_matrix} and properties of the Hilbert metric, the following inequalities hold: 
\begin{align*}
    \hmet\Big(\ind{\vect{u}}{n+1, 1}, \ind{\opt{\vect{u}}}{1}\Big)&\leq D\cdot  \hmet\Big(\ind{{\vect{u}}}{n+1, 2}, \ind{\opt{\vect{u}}}{2}\Big),  &\hmet\Big(\ind{\vect{u}}{n+3, 1}, \ind{\opt{\vect{u}}}{3}\Big) \leq D\cdot  \hmet\Big(\ind{{\vect{u}}}{n+1, 2}, \ind{\opt{\vect{u}}}{2}\Big).
\end{align*}
Furthermore, we can also show that 
\begin{align*}
  \hmet\Big(\ind{\vect{u}}{n+1, 2}, \ind{\opt{\vect{u}}}{2}\Big) \leq D\cdot \max\Big(\hmet\big(\ind{\vect{u}}{n,1},  \ind{\opt{\vect{u}}}{1}\big), \hmet\big(\ind{\vect{u}}{n,3},  \ind{\opt{\vect{u}}}{3}\big)\Big). 
\end{align*}
Toghether with these three inequalities, we can prove the statement by the induction of $n$. 

See \S\ref{sec:proof_convergences} for the details. 
\end{proofs}

As a consequence, we can also show the convergence of the marginal distributions. 
\begin{proposition}
  \label{prop:dist_convergence_2}
  Let  $\big( (\ind{\vect{u}}{n, i})_{i\in \nset{3}} \big)_{n\in \nat}$ be a Sinkhorn iteration and 
 $\big((\ind{\vect{P}}{n, i})_{i\in \nset{2}}\big)_{n\in \nat}$ be the matrices induced by the iteration. 
  For each $n\in \nat\backslash \{0\}$, the following inequalities hold: 
  \begin{align*}
    &\hmet\Big(\vect{a},\ind{\vect{P}}{n, 1}\vect{1}_{m_2}\Big) = \hmet\Big(\vect{b},(\ind{\vect{P}}{n, 2})^{\top}\vect{1}_{m_2}\Big) = 0,\\
    &\hmet\Big((\ind{\vect{P}}{n, 1})^{\top}\vect{1}_{m_1}, \ind{\vect{P}}{n, 2}\vect{1}_{m_3}\Big) \leq 2d\big(1+D^2\big)D^{2n-1}.
  \end{align*}
  where $D \defeq \max \big(\lambda(\ind{\vect{K}}{1}), \lambda(\ind{\vect{K}}{2})\big)$ and $d \defeq \max\Big(\hmet\big(\ind{\vect{P}}{0, 1}, \ind{\opt{\vect{P}}}{1}\big), \hmet\big(\ind{\vect{P}}{0, 3}, \ind{\opt{\vect{P}}}{3}\big)\Big)$.
\end{proposition}
\begin{proofs}
The equalities are easy to prove, and the inequality follows from Thm.~\ref{thm:global_convergence_2}. See \S\ref{sec:proofDistConvergences} for the details. 
\end{proofs}
Importantly, the equalities $\big\|\ind{\vect{P}}{n, 1}\big\|_1 = \big\|\ind{\vect{P}}{n, 2}\big\|_1 = 1$ hold, thus $\vect{a} = \ind{\vect{P}}{n, 1}\vect{1}_{m_2}$ and $\vect{b} = (\ind{\vect{P}}{n, 2})^{\top}\vect{1}_{m_2}$ for any $n\in \nat\backslash\{0\}$. We  also note that the $L^1$ norm is bounded by the square root of the Hilbert metric as follows  
\begin{align*}
    \bigg\|(\ind{\vect{P}}{n, 1})^{\top}\vect{1}_{m_1} -  \ind{\vect{P}}{n, 2}\vect{1}_{m_3}\bigg\|_1 \leq \sqrt{2}\bigg(\hmet\Big((\ind{\vect{P}}{n, 1})^{\top}\vect{1}_{m_1}, \ind{\vect{P}}{n, 2}\vect{1}_{m_3}\Big)\bigg)^{1/2}
\end{align*}
by the Pinsker's inequality (e.g.~\cite{slivkins2019introduction}).

\subsection{The Case of $M \geq 2$}

For the general case $M \geq 2$, we present the following exponential convergence result; it include the case $M = 2$ with the worse convergence rate. 
\begin{theorem}[convergence]
  \label{thm:global_convergence}
  Let  $\big( (\ind{\vect{u}}{n, i})_{i\in \nset{M+1}} \big)_{n\in \nat}$ be a Sinkhorn iteration, and assume that given vectors $(\ind{\opt{\vect{u}}}{i})_{i\in \nset{M+1}}$ 
  become an optimal solution $\big(\ind{\opt{\vect{f}}}{i}\big)_{i\in \nset{M+1}}$ of the Lagrange dual problem defined in Prop.~\ref{prop:dual_regularized_seq_composed_optimal_transport} by translating $\ind{\opt{\vect{f}}}{i}\defeq \epsilon \log(\ind{\opt{\vect{u}}}{i})$ for each $i\in \nset{M+1}$. 
  For each $n\in \nat$, the following inequalities hold: 
  \begin{align*}
    \hmet\Big(\ind{\vect{u}}{n, i}, \ind{\opt{\vect{u}}}{i}\Big) \leq E^{n}\cdot e \quad \text{for any $i\in \nset{M+1}$}
  \end{align*}
  where $E \defeq \max_{i\in \nset{M}} \Big\{\lambda\big(\ind{\vect{K}}{i}\big)\Big\}$ and $e \defeq \max_{i\in \nset{M+1}}\Big\{\hmet\big(\ind{\vect{u}}{0, i}, \ind{\opt{\vect{u}}}{i}\big)\Big\}$. 
\end{theorem}
\begin{proofs}
    Similar to the proof of Thm.~\ref{thm:global_convergence_2}, for each $i\in \nset{2, M}$ we obtain the following inequality:
    \begin{align*}
  \hmet\Big(\ind{\vect{u}}{n+1, i}, \ind{\opt{\vect{u}}}{i}\Big)\leq E\cdot \max\bigg(\hmet\Big(\ind{\vect{u}}{n, i-1}, \ind{\opt{\vect{u}}}{i-1}\Big),\, \hmet\Big(\ind{\vect{u}}{n, i+1}, \ind{\opt{\vect{u}}}{i+1}\Big)\bigg).
\end{align*}
We can also show that 
\begin{align*}
  \hmet\Big(\ind{\vect{u}}{n+1, 1}, \ind{\opt{\vect{u}}}{1}\Big) &\leq E\cdot \hmet\Big(\ind{\vect{u}}{n+1, 2}, \ind{\opt{\vect{u}}}{2}\Big)\\
  \hmet\Big(\ind{\vect{u}}{n+1, 1}, \ind{\opt{\vect{u}}}{1}\Big) &\leq E\cdot \hmet\Big(\ind{\vect{u}}{n+1, M}, \ind{\opt{\vect{u}}}{M}\Big). 
\end{align*}
Since $E^2 < E$, we can show the desired inequalities by the induction of $n$. See \S\ref{sec:proof_convergences} for the details.  
\end{proofs}

The convergence of the marginal distributions also holds. We note that $\big\|\ind{\vect{P}}{n, i}\big\|_1 \not = 1$ for each $i\in \nset{2, M-1}$ in general, while $\big\|\ind{\vect{P}}{n, 1}\big\|_1 = \big\|\ind{\vect{P}}{n, M}\big\|_1 = 1$ hold. In fact, $\big\|\ind{\vect{P}}{n, i}\big\|_1$ can be strictly smaller than $1$, which is a notable difference from the case of $M=2$. 
\begin{proposition}
  \label{prop:dist_convergence}
  Let  $\big( (\ind{\vect{u}}{n, i})_{i\in \nset{M+1}} \big)_{n\in \nat}$ be a Sinkhorn iteration and 
 $\big((\ind{\vect{P}}{n, i})_{i\in \nset{M}}\big)_{n\in \nat}$ be the matrices induced by the iteration. 
  For each $n\in \nat\backslash \{0\}$, the following inequalities hold: 
  \begin{align*}
    &\hmet\Big(\vect{a},\ind{\vect{P}}{n, 1}\vect{1}_{m_{2}}\Big) = \hmet\Big(\vect{b},(\ind{\vect{P}}{n, M})^{\top}\vect{1}_{m_{M}}\Big) = 0\\
    &\hmet\Big((\ind{\vect{P}}{n, i})^{\top}\vect{1}_{m_i}, \ind{\vect{P}}{n, i+1}\vect{1}_{m_{i+2}}\Big) \leq 2e(1+E)E^{n-1} \quad \text{for any $i\in \nset{1, M-1}$}.
  \end{align*}
  where $E \defeq \max_{i\in \nset{M}} \Big\{\lambda\big(\ind{\vect{K}}{i}\big)\Big\}$ and $e \defeq \max_{i\in \nset{M+1}}\Big\{\hmet\big(\ind{\vect{u}}{0, i}, \ind{\opt{\vect{u}}}{i}\big)\Big\}$. 
\end{proposition}
The proof is similar to that of Prop.~\ref{prop:dist_convergence_2}: see \S\ref{sec:proofDistConvergences} for the details. 


\section{Complexity when $M = 2$}
\label{sec:time_complexity}
In this section, we fix $M = 2$ and present a novel worst-case complexity analysis of the Sinkhorn algorithm. All proofs of the statements presented in this section is in $\S$\ref{sec:proofTimeComp}.  For the case $M=2$, we recall the Sinkhorn iteration is given by 
\begin{align*}
  \ind{\vect{u}}{n+1, 2} &\defeq \Bigg(\frac{(\ind{\vect{K}}{1})^{\top} \ind{\vect{u}}{n,1}}{\ind{\vect{K}}{2}\ind{\vect{u}}{n,3}}\Bigg)^{1/2},\\
  \ind{\vect{u}}{n+1, 1} &\defeq \frac{\vect{a}}{\ind{\vect{K}}{1}\frac{\vect{1}_{m_{2}}}{\ind{\vect{u}}{n+1,2}}}, &\ind{\vect{u}}{n+1, 3} \defeq \frac{\vect{b}}{(\ind{\vect{K}}{2})^{\top}\ind{\vect{u}}{n+1,2}}.
\end{align*}
In this section, we fix the initial vectors $\big(\ind{\vect{u}}{0, i}\big)_{i\in \nset{3}}$ that are defined by  $\ind{\vect{u}}{0, 1}\defeq \vect{1}_{m_1}/\big\|\ind{\vect{K}}{1}\big\|_1$, $\ind{\vect{u}}{0, 2}\defeq \vect{1}_{m_2}$,  and $\ind{\vect{u}}{0, 3}\defeq \vect{1}_{m_3}/\big\|\ind{\vect{K}}{2}\big\|_1$.
We write $\laglange\Big(\big(\ind{\vect{u}}{n, i}\big)_{i\in \nset{3}}\Big)$ for the Lagrangian that is defined in Prop.~\ref{prop:dual_regularized_seq_composed_optimal_transport} with substituting $\ind{\vect{f}}{i}\defeq \epsilon \log\big(\ind{\vect{u}}{n, i}\big) $ for any $i\in \nset{3}$.

We first recall that the matrices $\ind{\vect{P}}{n, 1}, \ind{\vect{P}}{n, 2}$  always satisfy the constraints of the edges with $\vect{a}$ and $\vect{b}$. This is an important property that we rely on the most of proofs in this section. 
\begin{lemma}
  \label{lem:transportmat}
  For any $n\in \nat$, the matrices $\ind{\vect{P}}{n, 1}, \ind{\vect{P}}{n, 2}$ satisfy the following equalities 
  \begin{align*}
   \Big\|\ind{\vect{P}}{n, 1}\Big\|_1 \defeq \sum^{m_1}_{j=1}\sum^{m_{2}}_{k= 1}\ind{P}{n, 1}_{jk} = 1, &&   \Big\|\ind{\vect{P}}{n, 2}\Big\|_1 \defeq\sum^{m_2}_{j=1}\sum^{m_{3}}_{k= 1}\ind{P}{n, 2}_{jk} = 1.
  \end{align*} 
  In particular, the equalities $\vect{a} = \ind{\vect{P}}{1}\vect{1}_{m_2}$ and $\vect{b} = {\ind{\vect{P}}{2}}^\top \vect{1}_{m_2}$ hold when $n > 0$. 
\end{lemma}

We then show the next key lemma that characterize the difference $\laglange\big((\ind{\vect{u}}{n+1, i})_{i\in \nset{3}}\big) - \laglange\big((\ind{\vect{u}}{n, i})_{i\in \nset{3}}\big)$ in terms of the Kullback-Leibler divergence. 

\begin{lemma}
  \label{lem:diffLagrange}
  For each $n\in \nat$, the following equality holds: 
  \begin{align*}
    &\laglange\big((\ind{\vect{u}}{n+1, i})_{i\in \nset{3}}\big) - \laglange\big((\ind{\vect{u}}{n, i})_{i\in \nset{3}}\big)\\
     = &\epsilon\Bigg( \kld[\bigg]{\vect{a}}{\ind{\vect{P}}{n, 1}\Big( \frac{\ind{\vect{P}}{n,2}\vect{1}_{m_3}}{\vect{1}_{m_1}^{\top}\ind{\vect{P}}{n, 1}}\Big)^{1/2}} +  \kld[\bigg]{\vect{b}}{(\ind{\vect{P}}{n, 2})^{\top}\Big( \frac{\vect{1}_{m_1}^{\top}\ind{\vect{P}}{n, 1}}{\ind{\vect{P}}{n,2}\vect{1}_{m_3}}\Big)^{1/2}}\Bigg),
  \end{align*}
  where the Kullback-Leibler divergence $\kld{\vect{c}}{\vect{d}}$ is defined by $\kld{\vect{c}}{\vect{d}}\defeq \vect{c}^{\top}\log\big(\frac{\vect{c}}{\vect{d}}\big)$. 

  Furthermore, the inequality $\laglange\big((\ind{\vect{u}}{n+1, i})_{i\in \nset{3}}\big) - \laglange\big((\ind{\vect{u}}{n, i})_{i\in \nset{3}}\big) \geq 0$ holds. 
\end{lemma}

Our proof strategy is focusing on the distances betweeen (i) $\vect{a}$ and $\ind{\vect{P}}{n, 1}\Big( \frac{\ind{\vect{P}}{n,2}\vect{1}_{m_3}}{\vect{1}_{m_1}^{\top}\ind{\vect{P}}{n, 1}}\Big)^{1/2}$; and (ii) $\vect{b}$ and $(\ind{\vect{P}}{n, 2})^{\top}\Big( \frac{\vect{1}_{m_1}^{\top}\ind{\vect{P}}{n, 1}}{\ind{\vect{P}}{n,2}\vect{1}_{m_3}}\Big)^{1/2}$. In fact, there is the following intuition: for each iteration step $n$, we further update only 
$\ind{\vect{u}}{n+1, 2}$, and construct matrices. These matrices satisfy the constraint on the boundary, while they do not satisfy the constraint on the edges with $\vect{a}$ and $\vect{b}$: their marginal (sub)distributions are precisely $\ind{\vect{P}}{n, 1}\Big( \frac{\ind{\vect{P}}{n,2}\vect{1}_{m_3}}{\vect{1}_{m_1}^{\top}\ind{\vect{P}}{n, 1}}\Big)^{1/2}$ and $(\ind{\vect{P}}{n, 2})^{\top}\Big( \frac{\vect{1}_{m_1}^{\top}\ind{\vect{P}}{n, 1}}{\ind{\vect{P}}{n,2}\vect{1}_{m_3}}\Big)^{1/2}$, respectively. We formulate this fact as follows.

\begin{proposition}
\label{prop:chHalf}
For each $n\in \nat$, let $\ind{\vect{P}'}{n, 1}$ and $\ind{\vect{P}'}{n, 2}$ be the matrices given by 
\begin{align*}
    \ind{\vect{P}'}{n, 1} &\defeq \diag{\ind{\vect{u}}{n, 1}}\ind{\vect{K}}{1}\diag{\frac{\vect{1}_{m_{2}}}{\ind{\vect{u}}{n+1, 2}}},\\
    \ind{\vect{P}'}{n, 2} &\defeq \diag{\ind{\vect{u}}{n+1, 2}}\ind{\vect{K}}{2}\diag{\ind{\vect{u}}{n, 3}}, 
\end{align*}
The following properties hold: 
\begin{itemize}
    \item they are consistent on the boundary, i.e.,  ${\ind{\vect{P}'}{n , 1}}^{\top}\vect{1}_{m_1} = {\ind{\vect{P}'}{n, 2}}\vect{1}_{m_3}$,
    \item on the edges, the following equalities hold: 
    \begin{align*}
        \ind{\vect{P}'}{n, 1}\vect{1}_{m_2} = \ind{\vect{P}}{n, 1}\bigg( \frac{\ind{\vect{P}}{n,2}\vect{1}_{m_3}}{\vect{1}_{m_1}^{\top}\ind{\vect{P}}{n, 1}}\bigg)^{1/2},\quad {\ind{\vect{P}'}{n, 2}}^{\top}\vect{1}_{m_2} = (\ind{\vect{P}}{n, 2})^{\top}\bigg( \frac{\vect{1}_{m_1}^{\top}\ind{\vect{P}}{n, 1}}{\ind{\vect{P}}{n,2}\vect{1}_{m_3}}\bigg)^{1/2}.
    \end{align*}
\end{itemize}
We note that $\Big\|\ind{\vect{P}'}{n, 1}\Big\|_1\leq 1$ and $\Big\|\ind{\vect{P}'}{n, 2}\Big\|_1\leq 1$, and they are strictly smaller than $1$ in general. 

\end{proposition}

From the characterization shown in Prop.~\ref{prop:chHalf}, the distance $\laglange\big((\ind{\vect{u}}{n+1, i})_{i\in \nset{3}}\big) - \laglange\big((\ind{\vect{u}}{n, i})_{i\in \nset{3}}\big)$ can be represented by the errors between (i) $\vect{a}$ and $\ind{\vect{P}'}{n, 1}\vect{1}_{m_2}$; and (ii) $\vect{b}$ and ${\ind{\vect{P}'}{n, 2}}^{\top}\vect{1}_{m_2}$.

We move on to the evaluation of how far the initial value $\laglange\big((\ind{\vect{u}}{0, i})_{i\in \nset{3}}\big)$ is from the optimal value $\laglange\big((\ind{\opt{\vect{u}}}{i})_{i\in \nset{3}}\big)$. 
\begin{lemma}
  \label{lem:maxdistanceLagrange}
  Let $K\defeq \min_{jl} \sum_k \ind{K}{1}_{jk}\ind{K}{2}_{kl}$. The following inequality holds: 
  \begin{align*}
    \laglange\big((\ind{\opt{\vect{u}}}{i})_{i\in \nset{3}}\big) - \laglange\big((\ind{\vect{u}}{0, i})_{i\in \nset{3}}\big)\leq \epsilon \log\Bigg(\frac{\big\|\ind{\vect{K}}{1}\big\|_1\big\|\ind{\vect{K}}{2}\big\|_1}{K}\Bigg)
  \end{align*}
where the vectors $(\ind{\opt{\vect{u}}}{i})_{i\in \nset{3}}$ are optimal (defined in Thm.~\ref{thm:global_convergence_2}). 
\end{lemma}

Together with~\cref{lem:diffLagrange} and~\cref{lem:maxdistanceLagrange}, we present an upper bound of the number of iterations to terminate with respect to the errors between $\vect{a}$ and $\ind{\vect{P}'}{n, 1}\vect{1}_{m_2}$, and $\vect{b}$ and ${\ind{\vect{P}'}{n, 2}}^{\top}\vect{1}_{m_2}$.
\begin{lemma}
  \label{lem:terminating}
  Let $K\defeq \min_{jl} \sum_k \ind{K}{1}_{jk}\ind{K}{2}_{kl}$. Given a constant $\delta > 0$,  assume the following stopping criterion: 
  \begin{align*}
    \Bigg\|\vect{a} - \ind{\vect{P}}{n, 1}\bigg( \frac{\ind{\vect{P}}{n,2}\vect{1}_{m_3}}{\vect{1}_{m_1}^{\top}\ind{\vect{P}}{n, 1}}\bigg)^{1/2}\Bigg\|_1 + \Bigg\|\vect{b} - (\ind{\vect{P}}{n, 2})^{\top}\bigg( \frac{\vect{1}_{m_1}^{\top}\ind{\vect{P}}{n, 1}}{\ind{\vect{P}}{n,2}\vect{1}_{m_3}}\bigg)^{1/2}\Bigg\|_1\leq \delta.
  \end{align*}
  Then, the Sinkhorn iteration terminates in at most 
  \begin{align*}
    1 + \frac{4}{\delta^2}\log\Bigg(\frac{\big\|\ind{\vect{K}}{1}\big\|_1\big\|\ind{\vect{K}}{2}\big\|_1}{K}\Bigg) \quad \text{ times.}
  \end{align*}
\end{lemma}
\begin{proofs}
    It follows from Lem.~\ref{lem:diffLagrange} and Lem.~\ref{lem:maxdistanceLagrange}. We also use the Pinsker's inequality (e.g.~\cite{slivkins2019introduction}): the inequality $\|P-Q\|^2_1\leq 2\kld{P}{Q}$ for any distribution $P$ and subdistribution $Q$. See \S\ref{sec:proofTimeComp} for the details.  
\end{proofs}

Finally, we prove our main theorem---an analysis of its worst case complexity with respect to arithmetic operations---with the results we have prepared. We recall a known technique presented in~\cite{AltschulerWR17}  to obtain a feasible solution with an error bound. 
\begin{lemma}[\cite{AltschulerWR17}]
\label{lem:getFeasibleSolution}
Let $\vect{P}\in \Rnnegmat{m}{n}$, $\vect{a}\in \setdist{m}$, and $\vect{b}\in \setdist{n}$. 
There is a procedure that constructs an transportation plan $\widetilde{\vect{P}}\in \Rnnegmat{m}{n}$ for $\vect{a}$ and $\vect{b}$, that is, 
$\widetilde{\vect{P}}\vect{1}_{n} = \vect{a}$ and  ${\widetilde{\vect{P}}}^{\top}\vect{1}_{m} = \vect{b}$. 
Its worst-case complexity is $O(mn)$ and the error is bounded by 
\begin{align*}
  \big\|\vect{P}- \widetilde{\vect{P}}\big\|_1 \leq 2\Big(\big\|\vect{P}\vect{1}_{n} - \vect{a}\big\|_1 + \big\|{\vect{P}}^{\top}\vect{1}_{m} - \vect{b}\big\|_1\Big). 
\end{align*} 
  
\end{lemma}



In the proof, we construct a feasible solution from the following matrices $\Big((\ind{\vect{P}}{n+0.5, i})_{i\in \nset{2}}\Big)_{n\in \nat}$, which are normalization of $\ind{\vect{P}'}{n, 1}$ and $\ind{\vect{P}'}{n, 2}$ that are presented in Prop.~\ref{prop:chHalf}. 
\begin{definition}
  \label{def:halfPlan}
  Given a Sinkhorn iteration $\big( (\ind{\vect{u}}{n, i} \in \Rpos^{m_{i}})_{i\in \nset{3}} \big)_{n\in \nat}$, the \emph{matrices $\Big((\ind{\vect{P}}{n+0.5, i})_{i\in \nset{2}}\Big)_{n\in \nat}$ induced by the iteration} are defined as follows: 
  \begin{align*}
    \ind{\vect{P}}{n +0.5, 1} &\defeq \diag{\frac{\ind{\vect{u}}{n, 1}}{P}}\ind{\vect{K}}{1}\diag{\frac{\vect{1}_{m_{2}}}{\ind{\vect{u}}{n+1, 2}}},\\
    \ind{\vect{P}}{n+0.5, 2} &\defeq \diag{\ind{\vect{u}}{n+1, 2}}\ind{\vect{K}}{2}\diag{\frac{\ind{\vect{u}}{n, 3}}{P}}, 
  \end{align*}
  for any $n\in \nat$, where
  \begin{align*}
      P \defeq \Bigg\|\diag{\ind{\vect{u}}{n, 1}}\ind{\vect{K}}{1}\diag{\frac{\vect{1}_{m_{2}}}{\ind{\vect{u}}{n+1, 2}}}\Bigg\|_1 = \Bigg\| \diag{\ind{\vect{u}}{n+1, 2}}\ind{\vect{K}}{2}\diag{\ind{\vect{u}}{n, 3}}\Bigg\|_1.
  \end{align*}

\end{definition}
Note that they satisfy the condition on the boundary, that is, the equality ${\ind{\vect{P}}{n +0.5, 1}}^{\top}\vect{1}_{m_1} = {\ind{\vect{P}}{n +0.5, 3}}\vect{1}_{m_3}$ holds.

\begin{theorem}
  \label{thm:main_theorem}
  Let $\delta > 0$ and $\epsilon \defeq \frac{\delta}{2\log\big(m_1(m_2)^2m_3\big)}$. 
  We assume that the stopping criterion for the Sinkhorn iteration is given by 
  \begin{align*}
    &\Bigg\|\vect{a} - \ind{\vect{P}}{n, 1}\bigg( \frac{\ind{\vect{P}}{n,2}\vect{1}_{m_3}}{\vect{1}_{m_1}^{\top}\ind{\vect{P}}{n, 1}}\bigg)^{1/2}\Bigg\|_1 + \Bigg\|\vect{b} - (\ind{\vect{P}}{n, 2})^{\top}\bigg( \frac{\vect{1}_{m_1}^{\top}\ind{\vect{P}}{n, 1}}{\ind{\vect{P}}{n,2}\vect{1}_{m_3}}\bigg)^{1/2}\Bigg\|_1\\
    &\leq \frac{\delta}{16\max\Big(\big\|\ind{\vect{C}}{1}\big\|_{\infty}, \big\|\ind{\vect{C}}{2}\big\|_{\infty}\Big)},
  \end{align*}
  and assume that it terminates in $k$ step (it always terminates by Lem.~\ref{lem:terminating}). 
  The following inequality holds: 
  \begin{align*}
    \Big\langle\ind{\vect{C}}{1}, \widetilde{\ind{\vect{P}}{k+0.5, 1}} \Big\rangle + \Big\langle\ind{\vect{C}}{2}, \widetilde{\ind{\vect{P}}{k+0.5, 2}} \Big\rangle \leq \min_{\ind{\vect{P}}{1}, \ind{\vect{P}}{2}} \Big\langle\ind{\vect{C}}{1}, \ind{\vect{P}}{1} \Big\rangle + \Big\langle\ind{\vect{C}}{2}, \ind{\vect{P}}{2} \Big\rangle + \delta,
  \end{align*}
  where the tranportation plans $\widetilde{\ind{\vect{P}}{k+0.5, 1}}$ and $\widetilde{\ind{\vect{P}}{k+0.5, 2}}$ are obtained by Lem.~\ref{lem:getFeasibleSolution} from $\ind{\vect{P}}{k+0.5, 1}$ and $\ind{\vect{P}}{k+0.5, 2}$ that are defined in Def.~\ref{def:halfPlan}.
  Let $\vect{C}\in\Rnnegmat{m_1}{m_2} $ be defined by $C_{ij} \defeq \min_{k\in \nset{m_2}}\ind{C}{1}_{ik} + \ind{C}{2}_{kj}$. 
  Assuming that $\delta < \|\vect{C}\|_{\infty}$, 
  its worst-case complexity of the number of arithmetic operations\footnote{We count taking the square root $(\_)^{1/2}$ as an arithmetic operation.} is 
  \begin{align*}
    O\Bigg(\max\big( m_1m_2, m_2m_3\big) \Big(\max\big(\|\ind{\vect{C}}{1}\|_{\infty}, \|\ind{\vect{C}}{2}\|_{\infty}\big)\Big)^2 \log\big(m_1(m_2)^2m_3\big)\delta^{-3} \|\vect{C}\|_{\infty} \Bigg).
  \end{align*}
\end{theorem}

\paragraph{Another Stopping Criterion.}
You may ask that the error $\big\|\ind{\vect{P}}{n, 2}\vect{1}_{m_3}- (\ind{\vect{P}}{n, 1})^{\top}\vect{1}_{m_1}\big\|_1$ can be used for the stopping criterion instead of the one used in Thm.~\ref{thm:main_theorem}. The answer is yes: In fact, it only requires at most one additional iteration to terminate. Formally, the following two inequalities ensure that these two stopping criteria are essentially equivalent. 

\begin{lemma}
  \label{lem:diff1}
  For any $n\in \nat$, the following inequality holds: 
  \begin{align*}
    &\Bigg\|\vect{a} - \ind{\vect{P}}{n, 1}\bigg( \frac{\ind{\vect{P}}{n,2}\vect{1}_{m_3}}{\vect{1}_{m_1}^{\top}\ind{\vect{P}}{n, 1}}\bigg)^{1/2}\Bigg\|_1 + \Bigg\|\vect{b} - (\ind{\vect{P}}{n, 2})^{\top}\bigg( \frac{\vect{1}_{m_1}^{\top}\ind{\vect{P}}{n, 1}}{\ind{\vect{P}}{n,2}\vect{1}_{m_3}}\bigg)^{1/2}\Bigg\|_1\\
    \leq\ &\Big\|\ind{\vect{P}}{n, 2}\vect{1}_{m_3}- (\ind{\vect{P}}{n, 1})^{\top}\vect{1}_{m_1}\Big\|_1
  \end{align*}
\end{lemma}

\begin{lemma}
  \label{lem:diff2}
  For any $n\in \nat$, the following inequality holds: 
  \begin{align*}
    &\Big\|\ind{\vect{P}}{n+1, 2}\vect{1}_{m_3}- (\ind{\vect{P}}{n+1, 1})^{\top}\vect{1}_{m_1}\Big\|_1 \\
    \leq &\Bigg\|\vect{a} - \ind{\vect{P}}{n, 1}\bigg( \frac{\ind{\vect{P}}{n,2}\vect{1}_{m_3}}{\vect{1}_{m_1}^{\top}\ind{\vect{P}}{n, 1}}\bigg)^{1/2}\Bigg\|_1 + \Bigg\|\vect{b} - (\ind{\vect{P}}{n, 2})^{\top}\bigg( \frac{\vect{1}_{m_1}^{\top}\ind{\vect{P}}{n, 1}}{\ind{\vect{P}}{n,2}\vect{1}_{m_3}}\bigg)^{1/2}\Bigg\|_1
  \end{align*}
\end{lemma}

\section{Related works}\label{sec:related}
\paragraph{Regularization of OTs:}
The algorithms and theoretical analyses for regularized (vanilla) OT extend beyond the Sinkhorn algorithm and its convergence and complexity results.
For example, algorithms for entropy-regularized OT problems include not only the Sinkhorn algorithm but also methods based on continuous optimization \citep{aistats/LuoXH23} and ordinary differential equations (ODEs) \citep{Hiew24}.
In addition, the authors in \citep{DBLP:conf/colt/Weed18,DBLP:journals/corr/abs-2309-11666} have performed error analyses between the minimizers of entropy or general Bregman divergence-regularized problems and those of unregularized problems.
In this paper, we focus on the Sinkhorn algorithm for SeqOT, investigating its convergence and complexity with a view towards applications in hierarchical planning.
\paragraph{Analysis for the Sinkhorn algorithms:}
In addition to the complexity analysis of the Sinkhorn algorithm for vanilla OT in \citep{AltschulerWR17,jmlr/LinHJ22}, progress has been made in analyzing the convergence and complexity of the Sinkhorn algorithm for various variants of OTs.
These include OT for probability distributions on continuous spaces \citep{colt/GrecoNCD23}, semi-relaxed OT \citep{Fukunaga22}, inequality-constrained OT \citep{DBLP:journals/corr/abs-2403-05054}, and multi-marginal (partial) OT \citep{LinHCJ22,DBLP:conf/icml/KosticSP22,DBLP:journals/siamjo/Carlier22,DBLP:conf/aistats/LeNN0H22}.
However, OT with compositional constraints, such as SeqOT, is not one of these variants.
Thus, the theoretical analysis of the Sinkhorn algorithm for SeqOT requires the development of new approaches that differ from those used in existing analyses.
%
%
%
\section{Conclusion}\label{sec:concllusion}
We formulate the entropic regularization of sequentially composed optimal transports, and introduce the Sinkhorn algorithm for them. We prove its exponential convergence to the optimal solution with respect to the Hilbert metric. For the case of $M = 2$, we provide a complexity analysis of the Sinkhorn algorithm. 

As a future work, we plan to evaluate its efficiency through numerical experiments, compared to the existing algorithm~\citep{WatanabeIsobe}. 
Proving a similar worst-case complexity result for the cases $M > 2$ is also interesting future work. Additionally, we plan to extend the Sinkhorn algorithm to string diagrams that have the sequential composition and the parallel composition. 


\bibliography{colt2025.bib}

\appendix


\section{Proof of Prop.~\ref{prop:dual_regularized_seq_composed_optimal_transport}}
\label{app:proof_prop_dual_regularized_seq_composed_optimal_transport}
In this section, we assume the setting of Def.~\ref{def:regularized_seq_composed_optimal_transport}.
\begin{definition}
  \label{def:lagrangian}
  The \emph{Lagrangian} $\laglange\Big( \big(\ind{\vect{P}}{i}\big)_{i\in \nset{M}}, \big(\ind{\vect{f}}{i}\big)_{i\in \nset{M+1}}\Big)$ is given by 
  \begin{align*}
    &\bigg(\sum_{i=1}^M \left\langle \ind{\vect{C}}{i}, \ind{\vect{P}}{i} \right\rangle -\epsilon\ent{\ind{\vect{P}}{i}} \bigg) + \left\langle\ind{\vect{f}}{1},\, \vect{a}- \ind{\vect{P}}{1}\vect{1}_{m_2} \right\rangle + \left\langle\ind{\vect{f}}{M+1},\, \vect{b}- {\ind{\vect{P}}{M}}^\top \vect{1}_{m_M}\right\rangle   \\
    &+ \sum_{i=2}^{M} \left\langle\ind{\vect{f}}{i}, \, {\ind{\vect{P}}{i-1}}^{\top}\vect{1}_{m_{i-1}}-\ind{\vect{P}}{i}\vect{1}_{m_{i+1}} \right\rangle. 
  \end{align*}
\end{definition}

We refer~\cite{boyd2004convex} as a reference for the Lagrange duality, including the definition of the Lagrangian. 

\begin{lemma}
  \label{lem:optimal_matrix}
  Assume that we fix the tuple $\big(\ind{\vect{f}}{i}\big)_{i\in \nset{M+1}}$. 
  Consider the matrices $\big(\ind{\opt{\vect{P}}}{i}\big)_{i\in \nset{M}}$ that are defined by 
  \begin{align*}
    \ind{\opt{P}}{i}_{jk} &\defeq \exp\Big( \big(\ind{f}{i}_j- \ind{f}{i+1}_k-\ind{C}{i}_{jk}\big)/\epsilon\Big) &&\quad\text{for all $i\in \nset{M-1}$, $j\in \nset{m_i}$, and $k\in \nset{m_{i+1}}$},\\
    \ind{\opt{P}}{M}_{jk} &\defeq \exp\Big( \big(\ind{f}{M}_j + \ind{f}{M+1}_k-\ind{C}{M}_{jk}\big)/\epsilon\Big) &&\quad\text{for all $j\in \nset{m_{M}}$, and $k\in \nset{m_{M+1}}$}.
  \end{align*}
  The matrices $\big(\ind{\opt{\vect{P}}}{i}\big)_{i\in \nset{M}}$ achieve the minimum of the Lagrangian as follows: 
  \begin{align*}
    \min_{\big(\ind{\vect{P}}{i}\big)_{i\in \nset{M}}} \laglange\Big( \big(\ind{\vect{P}}{i}\big)_{i\in \nset{M}}, \big(\ind{\vect{f}}{i}\big)_{i\in \nset{M+1}}\Big) =  \laglange\Big( \big(\ind{\opt{\vect{P}}}{i}\big)_{i\in \nset{M}}, \big(\ind{\vect{f}}{i}\big)_{i\in \nset{M+1}}\Big).
  \end{align*}
\end{lemma}
\begin{proof}
By the Karush--Kuhn--Tucker (KKT) theorem, a minimal point achieves the minimum. 
It thus suffices to solve the equations $\partial \laglange/\partial \ind{P}{i}_{jk} = 0 $ for all $i, j, k$. 
The derivatives are given by 
\begin{align*}
\frac{\partial \laglange}{\ind{P}{i}_{jk}} &= \ind{C}{i}_{jk} + \epsilon \log \ind{P}{i}_{jk} - \ind{f}{i}_j + \ind{f}{i+1}_k &&\quad\text{for all $i\in \nset{M-1}$, $j\in \nset{m_i}$, and $k\in \nset{m_{i+1}}$},\\
\frac{\partial \laglange}{\ind{P}{M}_{jk}} &= \ind{C}{M}_{jk} + \epsilon \log \ind{P}{M}_{jk} - \ind{f}{M}_j - \ind{f}{M+1}_k &&\quad\text{for all $j\in \nset{m_{M}}$, and $k\in \nset{m_{M+1}}$}.
\end{align*}
We now obtain the  matrices $\big(\ind{\opt{\vect{P}}}{i}\big)_{i\in \nset{M}}$ by solving the equations. 
\end{proof}

\paragraph{Proof of Prop.~\ref{prop:dual_regularized_seq_composed_optimal_transport}}
Since the primal problem satisfies the Slater's condition, the strong duality holds (see e.g.~\cite{boyd2004convex}).
Thus, it suffices to derive the dual problem by substituting the matrices $(\ind{\opt{\vect{P}}}{i})_{i\in \nset{M}}$ that are defined in Lem.~\ref{lem:optimal_matrix} into the Lagrangian.
For any $i\in \nset{M-1}$, the entropy $\ent{\ind{\opt{\vect{P}}}{i}}$ is given by 
\begin{align*}
  \ent{\ind{\opt{\vect{P}}}{i}}&\defeq-\sum_{j=1}^{m_i}\sum_{k=1}^{m_{i+1}} \ind{\opt{P}}{i}_{jk}\Big(\big(\log \ind{\opt{P}}{i}_{jk}\big) - 1\Big)\\
  &= -\sum_{j=1}^{m_i}\sum_{k=1}^{m_{i+1}} \Bigg(\frac{\ind{\opt{P}}{i}_{jk}}{\epsilon}\Big(\ind{f}{i}_j- \ind{f}{i+1}_k-\ind{C}{i}_{jk}\Big) - \ind{\opt{P}}{i}_{jk} \Bigg),
\end{align*}
and the entropy $\ent{\ind{\opt{\vect{P}}}{M}}$ is given by 
\begin{align*}
  \ent{\ind{\opt{\vect{P}}}{M}}&\defeq-\sum_{j=1}^{m_M}\sum_{k=1}^{m_{M+1}} \ind{\opt{P}}{M}_{jk}\Big(\big(\log \ind{\opt{P}}{M}_{jk}\big) - 1\Big)\\
  &= -\sum_{j=1}^{m_M}\sum_{k=1}^{m_{M+1}} \Bigg(\frac{\ind{\opt{P}}{M}_{jk}}{\epsilon}\Big(\ind{f}{M}_j + \ind{f}{M+1}_k-\ind{C}{M}_{jk}\Big) - \ind{\opt{P}}{M}_{jk} \Bigg).
\end{align*}
Therefore, we obtain the following equations: 
\begin{align*}
  &\laglange\Big( \big(\ind{\opt{\vect{P}}}{i}\big)_{i\in \nset{M}}, \big(\ind{\vect{f}}{i}\big)_{i\in \nset{M+1}}\Big) \\
  =& -\epsilon \sum^{M}_{i=1} \sum_{j=1}^{m_i}\sum_{k=1}^{m_{i+1}} \ind{\opt{P}}{i}_{jk} + \left\langle \ind{\vect{f}}{1}, \vect{a}\right\rangle + \left\langle \ind{\vect{f}}{M+1}, \vect{b}\right\rangle\\
  =&\left\langle \ind{\vect{f}}{1}, \vect{a}\right\rangle + \left\langle \ind{\vect{f}}{M+1}, \vect{b}\right\rangle - \epsilon \Bigg(\sum^{m_{M}}_{j=1}\sum^{m_{M+1}}_{k=1} \exp\Big(\big(\ind{f}{M}_j + \ind{f}{M+1}_k-\ind{C}{M}_{jk}\big)/\epsilon\Big) \\
  & + \sum_{i=1}^{M-1} \sum^{m_i}_{j=1}\sum^{m_{i+1}}_{k=1} \exp\Big( \big(\ind{f}{i}_j - \ind{f}{i+1}_k-\ind{C}{i}_{jl}\big)/\epsilon\Big) \Bigg). &\blacksquare
\end{align*}


\section{Proof of Lem.~\ref{lem:ch_optimal_matrix}}
\label{sec:proofChOTmat}
\paragraph{($\Rightarrow$)}
  Assume that the transportation plans $\big(\ind{\vect{P}}{i}\big)_{i\in \nset{M}}$ are optimal. 
  By the KKT theorem, the derivatives $\partial \laglange/\partial \ind{P}{i}_{jk}$ of the Lagrangian $L$ that is defined in Def.~\ref{def:lagrangian} becomes zero for all $i\in \nset{M}$, $j\in \nset{m_i}$, and $k\in \nset{m_{i+1}}$.
  We can thus that 
  \begin{align*}
    \ind{P}{i}_{jk} &= \exp\Big( \big(\ind{\opt{f}}{i}_j- \ind{\opt{f}}{i+1}_k-\ind{C}{i}_{jk}\big)/\epsilon\Big) &&\quad\text{for all $i\in \nset{M-1}$, $j\in \nset{m_i}$, and $k\in \nset{m_{i+1}}$},\\
    \ind{P}{M}_{jk} &= \exp\Big( \big(\ind{\opt{f}}{M}_j + \ind{\opt{f}}{M+1}_k-\ind{C}{M}_{jk}\big)/\epsilon\Big) &&\quad\text{for all $j\in \nset{m_{M}}$, and $k\in \nset{m_{M+1}}$},
  \end{align*}
  where the vectors $\big(\ind{\opt{\vect{f}}}{i}\big)_{i\in \nset{M+1}}$ are optimal with $\big(\ind{\vect{P}}{i}\big)_{i\in \nset{M}}$. 
  By defining 
  \begin{align*}
    \ind{\vect{x}}{i} &\defeq \exp\big(\ind{\vect{f}}{i}/\epsilon\big) &&\quad\text{for all $i\in \nset{M}$},\\
    \ind{\vect{x}}{M+1} &\defeq \exp\big(\ind{\vect{f}}{M+1}/\epsilon\big),
  \end{align*}
  we can conclude that 
  \begin{align*}
    \ind{\vect{P}}{i} &= \diag{\ind{\vect{x}}{i}}\ind{\vect{K}}{i}\diag{\frac{\vect{1}_{m_{i+1}}}{\ind{\vect{x}}{i+1}}} &\text{ for any $i\in \nset{M-1}$,}\\
    \ind{\vect{P}}{M} &= \diag{\ind{\vect{x}}{M}}\ind{\vect{K}}{M}\diag{\ind{\vect{x}}{M+1}}.
  \end{align*}
  \paragraph{($\Leftarrow$)}
  Assume that the vectors $\big(\ind{\vect{x}}{i}\big)_{i\in \nset{M}}$ and transportation plans $\big(\ind{\vect{P}}{i}\big)_{i\in \nset{M}}$ satisfy the following equations: 
  \begin{align*}
    \ind{\vect{P}}{i} &= \diag{\ind{\vect{x}}{i}}\ind{\vect{K}}{i}\diag{\frac{\vect{1}_{m_{i+1}}}{\ind{\vect{x}}{i+1}}} &\text{ for any $i\in \nset{M-1}$,}\\
    \ind{\vect{P}}{M} &= \diag{\ind{\vect{x}}{M}}\ind{\vect{K}}{M}\diag{\ind{\vect{x}}{M+1}}.
  \end{align*}
  By definining new vectors $\ind{\vect{f}}{i} \defeq \epsilon \log\big(\ind{\vect{x}}{i}\big)$ for all $i\in \nset{M+1}$, we can write the transportation plans $\big(\ind{\vect{P}}{i}\big)_{i\in \nset{M}}$ as follows:
  \begin{align*}
    \ind{P}{i}_{jk} &= \exp\Big( \big(\ind{f}{i}_j- \ind{f}{i+1}_k-\ind{C}{i}_{jk}\big)/\epsilon\Big) &&\quad\text{for all $i\in \nset{M-1}$, $j\in \nset{m_i}$, and $k\in \nset{m_{i+1}}$},\\
    \ind{P}{M}_{jk} &= \exp\Big( \big(\ind{f}{M}_j + \ind{f}{M+1}_k-\ind{C}{M}_{jk}\big)/\epsilon\Big) &&\quad\text{for all $j\in \nset{m_{M}}$, and $k\in \nset{m_{M+1}}$}.
  \end{align*}
  Then, we can see that the derivatives $\partial \laglange/\partial \ind{P}{i}_{jk}$ becomes zero for all $i\in \nset{M}$, $j\in \nset{m_i}$, and $k\in \nset{m_{i+1}}$, 
  and we can conclude that the transportation plans $\big(\ind{\vect{P}}{i}\big)_{i\in \nset{M}}$ are optimal by the KKT theorem. \qquad $\blacksquare$

\section{Proof of Thm.~\ref{thm:global_convergence_2} and Thm.~\ref{thm:global_convergence}}
\label{sec:proof_convergences}
\begin{lemma}
  \label{lem:hmet_constant}
  Let $\vect{v}, \vect{v}', \vect{c}\in \Rpos^{m}$. The following equation holds: 
  \begin{align*}
    &\hmet(\vect{c}\odot \vect{v}, \vect{c}\odot \vect{v}') = \hmet(\vect{v}, \vect{v}'). &&\blacksquare
  \end{align*}
\end{lemma}

\begin{lemma}
  \label{lem:hmet_sqrt}
  Let $\vect{v}, \vect{v}\in \Rpos^{m}$. The following equation holds: 
  \begin{align*}
    &\hmet\big(\vect{v}^{1/2}, {\vect{v}'}^{1/2}\big) = \frac{1}{2}\hmet(\vect{v}, \vect{v}'). &&\blacksquare
  \end{align*}
\end{lemma}

\begin{lemma}
  \label{lem:hmet_frac_triangle}
  Let $\vect{u}, \vect{v}, \vect{w}, \vect{x}\in \Rpos^{m}$. The following inequality holds: 
  \begin{align*}
    &\hmet\Big(\frac{\vect{u}}{\vect{x}}, \frac{\vect{v}}{\vect{w}}\Big) \leq \hmet(\vect{u}, \vect{v}) + \hmet(\vect{w}, \vect{x}). &&\blacksquare
  \end{align*}
\end{lemma}


\paragraph{Proof of Thm.~\ref{thm:global_convergence_2}}
By Cor.~\ref{cor:ch_optimal_matrix}, we obtain optimal transportation plans $\big(\ind{\opt{\vect{P}}}{i}\big)_{i\in \nset{2}}$ that are given by 
\begin{align*}
  \ind{\vect{P}}{1} &\defeq \diag{\ind{\opt{\vect{u}}}{1}}\ind{\vect{K}}{1}\diag{\frac{\vect{1}_{m_2}}{\ind{\opt{\vect{u}}}{2}}}, \quad \ind{\vect{P}}{2} \defeq \diag{\ind{\opt{\vect{u}}}{2}}\ind{\vect{K}}{2}\diag{\ind{\opt{\vect{u}}}{3}}.
\end{align*}

For each $n\in \nat$, we evaluate the distances $\hmet\Big(\ind{\vect{u}}{n+1, 1}, \ind{\opt{\vect{u}}}{1}\Big), \hmet\Big(\ind{\vect{u}}{n+1, 3}, \ind{\opt{\vect{u}}}{3}\Big), \hmet\Big(\ind{\vect{u}}{n+1, 2}, \ind{\opt{\vect{u}}}{2}\Big) $ in the following. 
\paragraph{Distances $\hmet\Big(\ind{\vect{u}}{n+1, 1}, \ind{\opt{\vect{u}}}{1}\Big)$ and $\hmet\Big(\ind{\vect{u}}{n+1, 3}, \ind{\opt{\vect{u}}}{3}\Big)$:}
Since $\ind{\opt{\vect{P}}}{1}$ satisfies the constraint $\ind{\opt{\vect{P}}}{1}\vect{1}_{m_2} = \vect{a}$, we have 
\begin{align*}
  \ind{\opt{\vect{P}}}{1}\vect{1}_{m_2} = \diag{\ind{\opt{\vect{u}}}{1}}\ind{\vect{K}}{1}\diag{\frac{\vect{1}_{m_2}}{\ind{\opt{\vect{u}}}{2}}}\vect{1}_{m_2} = \ind{\opt{\vect{u}}}{1}\odot \ind{\vect{K}}{1}\frac{\vect{1}_{m_2}}{\ind{\opt{\vect{u}}}{2}} = \vect{a}.
\end{align*}
Therefore, we know that 
\begin{equation}
  \label{eq:ch_optu1}
  \ind{\opt{\vect{u}}}{1} = \frac{\vect{a}}{ \ind{\vect{K}}{1}\frac{\vect{1}_{m_2}}{\ind{\opt{\vect{u}}}{2}}}. 
\end{equation}
Therefore, we can evaluate the distance $\hmet\Big(\ind{\vect{u}}{n+1, 1}, \ind{\opt{\vect{u}}}{1}\Big)$ as follows: 
\begin{align*}
  \hmet\Big(\ind{\vect{u}}{n+1, 1}, \ind{\opt{\vect{u}}}{1}\Big) &= \hmet\Bigg(\frac{\vect{a}}{ \ind{\vect{K}}{1}\frac{\vect{1}_{m_2}}{\ind{{\vect{u}}}{n+1, 2}}}, \frac{\vect{a}}{ \ind{\vect{K}}{1}\frac{\vect{1}_{m_2}}{\ind{\opt{\vect{u}}}{2}}}\Bigg) &&\text{by Eq.~\ref{eq:ch_optu1}}\\
  &= \hmet\Bigg(\frac{\vect{1}_{m_1}}{ \ind{\vect{K}}{1}\frac{\vect{1}_{m_2}}{\ind{{\vect{u}}}{n+1, 2}}}, \frac{\vect{1}_{m_1}}{ \ind{\vect{K}}{1}\frac{\vect{1}_{m_2}}{\ind{\opt{\vect{u}}}{2}}}\Bigg) &&\text{by Lem.~\ref{lem:hmet_constant} }\\
  &= \hmet\bigg(\ind{\vect{K}}{1}\frac{\vect{1}_{m_2}}{\ind{{\vect{u}}}{n+1, 2}},  \ind{\vect{K}}{1}\frac{\vect{1}_{m_2}}{\ind{\opt{\vect{u}}}{2}}\bigg) &&\text{by Def.~\ref{def:hmet} }\\
  &\leq \lambda\big(\ind{\vect{K}}{1}\big)\cdot \hmet\bigg(\frac{\vect{1}_{m_2}}{\ind{{\vect{u}}}{n+1, 2}},  \frac{\vect{1}_{m_2}}{\ind{\opt{\vect{u}}}{2}}\bigg) &&\text{by Lem.~\ref{lem:Birkhoff} }\\
  &= \lambda(\ind{\vect{K}}{1})\cdot \hmet\Big(\ind{{\vect{u}}}{n+1, 2}, \ind{\opt{\vect{u}}}{2}\Big) &&\text{by Def.~\ref{def:hmet} }\\
  &\leq D\cdot  \hmet\Big(\ind{{\vect{u}}}{n+1, 2}, \ind{\opt{\vect{u}}}{2}\Big) 
\end{align*}
Similarly, we can easily obtain the following inequality: 
\begin{align*}
  \hmet\Big(\ind{\vect{u}}{n+1, 3}, \ind{\opt{\vect{u}}}{3}\Big) \leq D\cdot  \hmet\Big(\ind{{\vect{u}}}{n+1, 2}, \ind{\opt{\vect{u}}}{2}\Big).
\end{align*}
\paragraph{Distance $\hmet\Big(\ind{\vect{u}}{n+1, 2}, \ind{\opt{\vect{u}}}{2}\Big)$:}
By Lem.~\ref{lem:hmet_sqrt}, we know that 
\begin{align*}
  \hmet\Big(\ind{\vect{u}}{n+1, 2}, \ind{\opt{\vect{u}}}{2}\Big) = \frac{1}{2}\hmet\Bigg(\frac{(\ind{\vect{K}}{1})^{\top} \ind{\vect{u}}{n,1}}{\ind{\vect{K}}{2}\ind{\vect{u}}{n,3}}, \frac{(\ind{\vect{K}}{1})^{\top} \ind{\opt{\vect{u}}}{1}}{\ind{\vect{K}}{2}\ind{\opt{\vect{u}}}{3}}\Bigg).
\end{align*}
By Lem.~\ref{lem:hmet_frac_triangle}, we have
\begin{align*}
  \hmet\Big(\ind{\vect{u}}{n+1, 2}, \ind{\opt{\vect{u}}}{2}\Big) \leq \frac{1}{2}\Bigg( \hmet\bigg((\ind{\vect{K}}{1})^{\top} \ind{\vect{u}}{n,1},  (\ind{\vect{K}}{1})^{\top} \ind{\opt{\vect{u}}}{1}\bigg) + \hmet\bigg(\ind{\vect{K}}{3}\ind{\vect{u}}{n,3}, \ind{\vect{K}}{3}\ind{\opt{\vect{u}}}{3}\bigg) \Bigg).
\end{align*}
By Lem.~\ref{lem:Birkhoff}, we can conclude that 
\begin{align*}
  \hmet\Big(\ind{\vect{u}}{n+1, 2}, \ind{\opt{\vect{u}}}{2}\Big) \leq D\cdot \max\bigg(\hmet\Big(\ind{\vect{u}}{n,1},  \ind{\opt{\vect{u}}}{1}\Big), \hmet\Big(\ind{\vect{u}}{n,3},  \ind{\opt{\vect{u}}}{3}\Big)\bigg). 
\end{align*}
Together with these inequalities, we can prove the statement by the induction of $n\in \nat\backslash \{0\}$. $\blacksquare$

\paragraph{Proof of Thm.~\ref{thm:global_convergence}}
By Cor.~\ref{cor:ch_optimal_matrix}, we obtain optimal transportation plans $(\ind{\opt{\vect{P}}}{i})_{i\in \nset{M}}$ that are given by 
\begin{align*}
  \ind{\vect{P}}{i} &\defeq \diag{\ind{\opt{\vect{u}}}{i}}\ind{\vect{K}}{i}\diag{\frac{\vect{1}_{m_2}}{\ind{\opt{\vect{u}}}{i}}}, \quad \ind{\vect{P}}{M} \defeq \diag{\ind{\opt{\vect{u}}}{M}}\ind{\vect{K}}{M}\diag{\ind{\opt{\vect{u}}}{M+1}},
\end{align*}
for any $i\in \nset{M-1}$.

For each $n\in \nat$, we evaluate the distances $\hmet\Big(\ind{\vect{u}}{n+1, i}, \ind{\opt{\vect{u}}}{i}\Big)$ for each $i\in \nset{M}$ in the following. 
\paragraph{Distances $\hmet\Big(\ind{\vect{u}}{n+1, 1}, \ind{\opt{\vect{u}}}{i}\Big)$ for any $i\in \nset{2, M-1}$:}
By Lem.~\ref{lem:hmet_sqrt}, we get  
\begin{align*}
  \hmet\Big(\ind{\vect{u}}{n+1, i}, \ind{\opt{\vect{u}}}{i}\Big) &= \frac{1}{2}\hmet\Bigg(\frac{(\ind{\vect{K}}{i-1})^{\top} \ind{\vect{u}}{n,i-1}}{\ind{\vect{K}}{i}\frac{\vect{1}_{m_{i+1}}}{\ind{\vect{u}}{n,i+1}}}, \frac{(\ind{\vect{K}}{i-1})^{\top} \ind{\opt{\vect{u}}}{i-1}}{\ind{\vect{K}}{i}\frac{\vect{1}_{m_{i+1}}}{\ind{\opt{\vect{u}}}{i+1}}}\Bigg).
\end{align*}
By Lem.~\ref{lem:hmet_frac_triangle} and Lem.~\ref{lem:Birkhoff}, we obtain the following inequality
\begin{align*}
  \hmet\Big(\ind{\vect{u}}{n+1, i}, \ind{\opt{\vect{u}}}{i}\Big)\leq E\cdot \max\bigg(\hmet\Big(\ind{\vect{u}}{n, i-1}, \ind{\opt{\vect{u}}}{i-1}\Big),\, \hmet\Big(\ind{\vect{u}}{n, i+1}, \ind{\opt{\vect{u}}}{i+1}\Big)\bigg).
\end{align*}
\paragraph{Distance $\hmet\Big(\ind{\vect{u}}{n+1, M}, \ind{\opt{\vect{u}}}{M}\Big)$:}
By Lem.~\ref{lem:hmet_sqrt}, we get  
\begin{align*}
  \hmet\Big(\ind{\vect{u}}{n+1, M}, \ind{\opt{\vect{u}}}{M}\Big) &= \frac{1}{2}\hmet\Bigg(\frac{(\ind{\vect{K}}{M-1})^{\top} \ind{\vect{u}}{n,M-1}}{\ind{\vect{K}}{M}{\ind{\vect{u}}{n,M+1}}}, \frac{(\ind{\vect{K}}{M-1})^{\top} \ind{\opt{\vect{u}}}{M-1}}{\ind{\vect{K}}{M}{\ind{\opt{\vect{u}}}{M+1}}}\Bigg).
\end{align*}
By Lem.~\ref{lem:hmet_frac_triangle} and Lem.~\ref{lem:Birkhoff}, we obtain the following inequality
\begin{align*}
  \hmet\Big(\ind{\vect{u}}{n+1, M}, \ind{\opt{\vect{u}}}{M}\Big)\leq E\cdot \max\bigg(\hmet\Big(\ind{\vect{u}}{n, M-1}, \ind{\opt{\vect{u}}}{M-1}\Big),\, \hmet\Big(\ind{\vect{u}}{n, M+1}, \ind{\opt{\vect{u}}}{M+1}\Big)\bigg).
\end{align*}
\paragraph{Distances $\hmet\Big(\ind{\vect{u}}{n+1, 1}, \ind{\opt{\vect{u}}}{1}\Big)$ and $\hmet\Big(\ind{\vect{u}}{n+1, M+1}, \ind{\opt{\vect{u}}}{M+1}\Big)$:}
By the same argument of the proof of Thm.~\ref{thm:global_convergence_2}, we can easily show that
\begin{align}
  \hmet\Big(\ind{\vect{u}}{n+1, 1}, \ind{\opt{\vect{u}}}{1}\Big) &\leq E\cdot \hmet\Big(\ind{\vect{u}}{n+1, 2}, \ind{\opt{\vect{u}}}{2}\Big)\\
  \hmet\Big(\ind{\vect{u}}{n+1, 1}, \ind{\opt{\vect{u}}}{1}\Big) &\leq E\cdot \hmet\Big(\ind{\vect{u}}{n+1, M}, \ind{\opt{\vect{u}}}{M}\Big)
\end{align}
Since $E^{n+2} < E^{n+1}$ (note that $0\leq E < 1$), we obtain the desired inequality by the induction of $n\in \nat$. $\blacksquare$  

\section{Proof of Prop.~\ref{prop:dist_convergence_2} and Prop.~\ref{prop:dist_convergence}}
\label{sec:proofDistConvergences}

\paragraph{Proof of Prop.~\ref{prop:dist_convergence_2}} 
For each $n\in \nat\backslash \{0\}$, the following equations hold: 
\begin{align*}
  \hmet\Big(\vect{a},\ind{\vect{P}}{n, 1}\vect{1}_{m_2}\Big) &= \hmet\bigg(\vect{a},\ind{\vect{u}}{n, 1}\odot \ind{\vect{K}}{1}\frac{\vect{1}_{m_2}}{\ind{\vect{u}}{n, 2}} \bigg) &&\text{by definition}\\
  & = \hmet\bigg(\vect{a},\frac{\vect{a}}{\ind{\vect{K}}{1}\frac{\vect{1}_{m_2}}{\ind{\vect{u}}{n, 2}}}\odot \ind{\vect{K}}{1}\frac{\vect{1}_{m_2}}{\ind{\vect{u}}{n, 2}} \bigg) &&\text{by definition}\\
  & = \hmet\bigg(\ind{\vect{K}}{1}\frac{\vect{1}_{m_2}}{\ind{\vect{u}}{n, 2}} , \vect{1}_{m_1}\odot \ind{\vect{K}}{1}\frac{\vect{1}_{m_2}}{\ind{\vect{u}}{n, 2}} \bigg)&&\text{by Lem.~\ref{lem:hmet_constant}}\\
  & = 0
\end{align*}
Similarly, we can easily show that $\hmet\Big(\vect{b},(\ind{\vect{P}}{n, 2})^{\top}\vect{1}_{m_{2}}\Big) = 0$.

For each  $n\in \nat\backslash\{0\}$, we can see that $ \hmet\Big((\ind{\vect{P}}{n, 1})^{\top}\vect{1}_{m_1}, \ind{\vect{P}}{n, 2}\vect{1}_{m_3}\Big) =  \frac{1}{2}\hmet\Big(\ind{\vect{u}}{n+1,2},  \ind{\vect{u}}{n,2}\Big)$ as follows:  
\begin{align*}
  \hmet\Big((\ind{\vect{P}}{n, 1})^{\top}\vect{1}_{m_1}, \ind{\vect{P}}{n, 2}\vect{1}_{m_3}\Big) &= \hmet\Big( (\ind{\vect{K}}{1})^{\top}\ind{\vect{u}}{n, 1}\odot \frac{\vect{1}_{m_2}}{\ind{\vect{u}}{n, 2}},  \ind{\vect{u}}{n,2} \odot \ind{\vect{K}}{2}\ind{\vect{u}}{n, 3}\Big) &&\text{by definition }\\
  &= 2\hmet\Bigg(\bigg(\frac{(\ind{\vect{K}}{1})^{\top}\ind{\vect{u}}{n, 1}}{\ind{\vect{K}}{2}\ind{\vect{u}}{n, 3}}\bigg)^{1/2},  \ind{\vect{u}}{n,2}\Bigg)&&\text{by Lem.~\ref{lem:hmet_constant} and Lem.~\ref{lem:hmet_sqrt}}\\
  &= 2\hmet\Big(\ind{\vect{u}}{n+1,2},  \ind{\vect{u}}{n,2}\Big)&&\text{by definition }
\end{align*}
Let $(\ind{\opt{\vect{u}}}{i})_{i\in \nset{3}}$ be vectors that become an optimal solution $\big(\ind{\opt{\vect{f}}}{i}\big)_{i\in \nset{3}}$ of the Lagrange dual problem defined in Prop.~\ref{prop:dual_regularized_seq_composed_optimal_transport} by translating $\ind{\opt{\vect{f}}}{i}\defeq \epsilon \log(\ind{\opt{\vect{u}}}{i})$ for each $i\in \nset{3}$. 
By the triangle inequality and Thm.~\ref{thm:global_convergence_2}, we obtain the following inequalities and conclude the proof: 
\begin{align*}
  \hmet\Big(\ind{\vect{u}}{n+1,2},  \ind{\vect{u}}{n,2}\Big) &\leq \hmet\Big(\ind{\vect{u}}{n+1,2},  \ind{\opt{\vect{u}}}{2}\Big)  + \hmet\Big(\ind{\opt{\vect{u}}}{2},  \ind{\vect{u}}{n,2}\Big) \\
  &\leq D^{2n+1}\cdot d + D^{2n-1}\cdot d\\
  &= d(1+D^2)D^{2n-1} &&\blacksquare
\end{align*}

\paragraph{Proof of Prop.~\ref{prop:dist_convergence}} 
For each $n\in \nat\backslash \{0\}$, we can easily prove that 
\begin{align*}
    \hmet\Big(\vect{a},\ind{\vect{P}}{n, 1}\vect{1}_{m_2}\Big) = \hmet\Big(\vect{b},(\ind{\vect{P}}{n, 2})^{\top}\vect{1}_{m_{2}}\Big) = 0
\end{align*}
 in the same way of the proof of Prop.~\ref{prop:dist_convergence_2}.

For each  $n\in \nat\backslash\{0\}$ and $i\in \nset{1, M-2}$, we can see that \[\hmet\Big((\ind{\vect{P}}{n, i})^{\top}\vect{1}_{m_i}, \ind{\vect{P}}{n, i+1}\vect{1}_{m_{i+2}}\Big) =  \frac{1}{2}\hmet\Big(\ind{\vect{u}}{n+1,i+1},  \ind{\vect{u}}{n,i+1}\Big)\] as follows:  
\begin{align*}
  &\hmet\Big((\ind{\vect{P}}{n, i})^{\top}\vect{1}_{m_i}, \ind{\vect{P}}{n, i+1}\vect{1}_{m_{i+2}}\Big)\\
  =& \hmet\Big( \big(\ind{\vect{K}}{i}\big)^{\top}\ind{\vect{u}}{n, i}\odot \frac{\vect{1}_{m_{i+1}}}{\ind{\vect{u}}{n, i+1}},  \ind{\vect{u}}{n,i+1} \odot \ind{\vect{K}}{i+1}\frac{\vect{1}_{m_{i+2}}}{\ind{\vect{u}}{n, i+2}}\Big) &&\text{by definition }\\
  =& 2\hmet\Bigg(\bigg(\frac{(\ind{\vect{K}}{i})^{\top}\ind{\vect{u}}{n, i}}{\ind{\vect{K}}{i+1}\frac{\vect{1}_{i+2}}{\ind{\vect{u}}{n, i+2}}}\bigg)^{1/2},  \ind{\vect{u}}{n,i+1}\Bigg)&&\text{by Lem.~\ref{lem:hmet_constant} and Lem.~\ref{lem:hmet_sqrt}}\\
  =& 2\hmet\Big(\ind{\vect{u}}{n+1,i+1},  \ind{\vect{u}}{n,i+1}\Big)&&\text{by definition }
\end{align*}
By the triangle inequality and Thm.~\ref{thm:global_convergence}, we obtain the following inequality: 
\begin{align*}
  \hmet\Big((\ind{\vect{P}}{n, i})^{\top}\vect{1}_{m_i}, \ind{\vect{P}}{n, i+1}\vect{1}_{m_{i+2}}\Big) &\leq 2e(1+E)E^n. 
\end{align*}
Similarly, we can prove that $\hmet\Big((\ind{\vect{P}}{n, M-1})^{\top}\vect{1}_{m_{M-1}}, \ind{\vect{P}}{n, M}\vect{1}_{m_{M+1}}\Big) \leq 2e(1+E)E^n$. $\blacksquare$

\section{Proofs in $\S$\ref{sec:time_complexity}}
\label{sec:proofTimeComp}

\subsection{Proof of Lem.~\ref{lem:transportmat}}
By definition of the initial vectors, 
\begin{align*}
  \sum_{jk} \ind{P}{0,1}_{jk} &= \sum_{jk} \ind{u}{0,1}_j \ind{K}{1}_{jk} (1/\ind{u}{0,2}_k) = \frac{1}{\big\|\ind{\vect{K}}{1}\big\|} \sum_{jk} \ind{K}{1}_{jk} = 1,\\
  \sum_{jk} \ind{P}{0,2}_{jk} &= \sum_{jk} \ind{u}{0,2}_j \ind{K}{2}_{jk} \ind{u}{0,3}_k = \frac{1}{\big\|\ind{\vect{K}}{2}\big\|} \sum_{jk} \ind{K}{2}_{jk} = 1.
\end{align*}

For any $n > 0$, the following equations hold: for any $j\in \nset{m_1}$,
\begin{align*}
  \sum_{k} \ind{P}{n,1}_{jk} &=  \Bigg(\frac{\vect{a}}{\ind{\vect{K}}{1}\frac{\vect{1}_{m_2}}{\ind{\vect{u}}{n, 2}}}\Bigg)_{j} \cdot \Bigg(\ind{\vect{K}}{1}\frac{\vect{1}_{m_2}}{\ind{\vect{u}}{n, 2}}\Bigg)_j = a_j,
\end{align*} 
and for any $k\in \nset{m_3}$, 
\begin{align*}
  \sum_{j} \ind{P}{n,1}_{jk} &=  \Bigg(\ind{\vect{u}}{n, 2}\ind{\vect{K}}{2}\Bigg)_k \cdot \Bigg(\frac{\vect{b}}{(\ind{\vect{K}}{2})^{\top}{\ind{\vect{u}}{n, 2}}}\Bigg)_{k} = b_k. \quad \blacksquare
\end{align*}

\subsection{Proof of Lem.~\ref{lem:diffLagrange}}
By definition, we get the following equality:
\begin{align*}
  &\laglange\big((\ind{\vect{u}}{n+1, i})_{i\in \nset{3}}\big) - \laglange\big((\ind{\vect{u}}{n, i})_{i\in \nset{3}}\big)\\
=& \epsilon \vect{a}^{\top}\log\Big(\ind{\vect{u}}{n+1, 1}\Big) + \epsilon \vect{b}^{\top}\log\Big(\ind{\vect{u}}{n+1, 3}\Big) + \Big(\sum_{jk} \ind{P_{jk}}{n+1, 1} \Big) +  \Big(\sum_{jk} \ind{P_{jk}}{n+1, 2} \Big) \\
-& \epsilon \vect{a}^{\top}\log\Big(\ind{\vect{u}}{n, 1}\Big) - \epsilon \vect{b}^{\top}\log\Big(\ind{\vect{u}}{n, 3}\Big) - \Big(\sum_{jk} \ind{P_{jk}}{n, 1} \Big) - \Big(\sum_{jk} \ind{P_{jk}}{n, 2} \Big).
\end{align*}
By Lem.~\ref{lem:transportmat}, we know that 
\begin{align*}
  &\laglange\big((\ind{\vect{u}}{n+1, i})_{i\in \nset{3}}\big) - \laglange\big((\ind{\vect{u}}{n, i})_{i\in \nset{3}}\big)\\
=&\epsilon \vect{a}^{\top}\log\Big(\ind{\vect{u}}{n+1, 1}\Big) + \epsilon \vect{b}^{\top}\log\Big(\ind{\vect{u}}{n+1, 3}\Big)
- \epsilon \vect{a}^{\top}\log\Big(\ind{\vect{u}}{n, 1}\Big) - \epsilon \vect{b}^{\top}\log\Big(\ind{\vect{u}}{n, 3}\Big)\\
=&\epsilon \vect{a}^{\top}\log\Bigg(\frac{\ind{\vect{u}}{n+1, 1}}{\ind{\vect{u}}{n, 1}}\Bigg) +\epsilon \vect{b}^{\top}\log\Bigg(\frac{\ind{\vect{u}}{n+1, 3}}{\ind{\vect{u}}{n, 3}}\Bigg).
\end{align*}
By definition of the Sinkhorn iteration,
\begin{align*}
  \ind{\vect{u}}{n+1, 1} &= \frac{\vect{a}}{\ind{\vect{K}}{1}\frac{\vect{1}_{m_{2}}}{\ind{\vect{u}}{n+1,2}}} = \frac{\vect{a}}{\ind{\vect{K}}{1}\Big(\frac{\ind{\vect{K}}{3}\ind{\vect{u}}{n,3}}{(\ind{\vect{K}}{1})^{\top} \ind{\vect{u}}{n,1}}\Big)^{1/2}}\\
  &= \frac{\vect{a}}{\ind{\vect{K}}{1}\diag{\frac{\vect{1}_{m_2}}{\ind{\vect{u}}{n, 2}}}\diag{\ind{\vect{u}}{n, 2}}\Big(\frac{\ind{\vect{K}}{3}\ind{\vect{u}}{n,3}}{(\ind{\vect{K}}{1})^{\top} \ind{\vect{u}}{n,1}}\Big)^{1/2}}\\
  &= \frac{\vect{a}}{\ind{\vect{K}}{1}\diag{\frac{\vect{1}_{m_2}}{\ind{\vect{u}}{n, 2}}}\bigg(\frac{\diag{\ind{\vect{u}}{n, 2}}\ind{\vect{K}}{3}\ind{\vect{u}}{n,3}}{\diag{\frac{\vect{1}_{m_2}}{\ind{\vect{u}}{n, 2}}}(\ind{\vect{K}}{1})^{\top} \ind{\vect{u}}{n,1}}\bigg)^{1/2}}\\
  &= \frac{\vect{a}}{\ind{\vect{K}}{1}\diag{\frac{\vect{1}_{m_2}}{\ind{\vect{u}}{n, 2}}}\Big(\frac{\ind{\vect{P}}{n,2}\vect{1}_{m_3}}{{\vect{1}_{m_1}}^{\top}\ind{\vect{P}}{n,1}} \Big)^{1/2}}
\end{align*}
Therefore, we can get the following equalities
\begin{align*}
  \frac{\ind{\vect{u}}{n+1, 1}}{\ind{\vect{u}}{n, 1}} 
  &= \frac{\vect{a}}{\diag{\ind{\vect{u}}{n, 1}}\ind{\vect{K}}{1}\diag{\frac{\vect{1}_{m_2}}{\ind{\vect{u}}{n, 2}}}\Big(\frac{\ind{\vect{P}}{n,2}\vect{1}_{m_3}}{{\vect{1}_{m_1}}^{\top}\ind{\vect{P}}{n,1}} \Big)^{1/2}}\\
  &= \frac{\vect{a}}{\ind{\vect{P}}{n, 1}\Big(\frac{\ind{\vect{P}}{n,2}\vect{1}_{m_3}}{{\vect{1}_{m_1}}^{\top}\ind{\vect{P}}{n,1}} \Big)^{1/2}}
\end{align*}
This concludes that  
\begin{align*}
  \epsilon \vect{a}^{\top}\log\bigg(\frac{\ind{\vect{u}}{n+1, 1}}{\ind{\vect{u}}{n, 1}}\bigg)  &= \epsilon \kld[\Bigg]{\vect{a}}{\ind{\vect{P}}{n, 1}\bigg( \frac{\ind{\vect{P}}{n,2}\vect{1}_{m_3}}{\vect{1}_{m_1}^{\top}\ind{\vect{P}}{n, 1}}\bigg)^{1/2}}. 
\end{align*}
Similarly, 
\begin{align*}
  \ind{\vect{u}}{n+1, 3} &=  \frac{\vect{b}}{(\ind{\vect{K}}{2})^{\top}\ind{\vect{u}}{n+1,2}} = \frac{\vect{b}}{(\ind{\vect{K}}{2})^{\top}\Big(\frac{(\ind{\vect{K}}{1})^{\top} \ind{\vect{u}}{n,1}}{\ind{\vect{K}}{2}\ind{\vect{u}}{n,3}}\Big)^{1/2}}\\
  &= \frac{\vect{b}}{(\ind{\vect{K}}{2})^{\top}\diag{\ind{\vect{u}}{n, 2}}\Big(\frac{\vect{1}_{m_1}^{\top}\ind{\vect{P}}{n, 1}}{\ind{\vect{P}}{n,2}\vect{1}_{m_3}}\Big)^{1/2}}\\
\end{align*}
Thus, we can conclude that 
\begin{align*}
  \frac{\ind{\vect{u}}{n+1, 3}}{\ind{\vect{u}}{n, 3}} &= \frac{\vect{b}}{(\ind{\vect{P}}{n, 2})^{\top}\Big(\frac{\vect{1}_{m_1}^{\top}\ind{\vect{P}}{n, 1}}{\ind{\vect{P}}{n,2}\vect{1}_{m_3}}\Big)^{1/2}}\\
  \epsilon \vect{b}^{\top}\log\bigg(\frac{\ind{\vect{u}}{n+1, 3}}{\ind{\vect{u}}{n, 3}}\bigg) &= \epsilon\kld[\Bigg]{\vect{b}}{(\ind{\vect{P}}{n, 2})^{\top}\bigg( \frac{\vect{1}_{m_1}^{\top}\ind{\vect{P}}{n, 1}}{\ind{\vect{P}}{n,2}\vect{1}_{m_3}}\bigg)^{1/2}}
\end{align*}
In order to prove the inequality  $\laglange\big((\ind{\vect{u}}{n+1, i})_{i\in \nset{3}}\big) - \laglange\big((\ind{\vect{u}}{n, i})_{i\in \nset{3}}\big) \geq 0$, 
it suffices to prove that $\ind{\vect{P}}{n, 1}\Big( \frac{\ind{\vect{P}}{n,2}\vect{1}_{m_3}}{\vect{1}_{m_1}^{\top}\ind{\vect{P}}{n, 1}}\Big)^{1/2}$ and $(\ind{\vect{P}}{n, 2})^{\top}\Big( \frac{\vect{1}_{m_1}^{\top}\ind{\vect{P}}{n, 1}}{\ind{\vect{P}}{n,2}\vect{1}_{m_3}}\Big)^{1/2}$ are subdistributions, 
since the Kullback-Leibler divergence $\kld{\vect{c}}{\vect{d}}$ is non-negative for any distribution $\vect{c}$ and subdistribution $\vect{d}$ (by the Gibbs' inequality). 
This can be easily proven by the H\"{o}lder's inequality as follows 
\begin{align*}
  \Bigg\|\ind{\vect{P}}{n, 1}\Big( \frac{\ind{\vect{P}}{n,2}\vect{1}_{m_3}}{\vect{1}_{m_1}^{\top}\ind{\vect{P}}{n, 1}}\Big)^{1/2}\Bigg\|_1 &= {\vect{1}_{m_1}}^{\top} \ind{\vect{P}}{n, 1}\bigg( \frac{\ind{\vect{P}}{n,2}\vect{1}_{m_3}}{\vect{1}_{m_1}^{\top}\ind{\vect{P}}{n, 1}}\bigg)^{1/2}\\
 &= \sum_{jk} \ind{P}{n, 1}_{jk} \bigg(\frac{\sum_l \ind{P}{n, 2}_{kl}}{\sum_{j'}\ind{P}{n, 1}_{j'k}}\bigg)^{1/2} \\
 &= \sum_{k} \Big(\sum_{j}\ind{P}{n, 1}_{jk}\Big)^{1/2} \Big(\sum_l \ind{P}{n, 2}_{kl}\Big)^{1/2} \\
 &\leq \Big(\sum_k \sum_j \ind{P}{n, 1}_{jk}\Big)^{1/2} \Big(\sum_k \sum_l \ind{P}{n, 2}_{kl}\Big)^{1/2}= 1 \quad \blacksquare
\end{align*}

\subsection{Proof of Prop.~\ref{prop:chHalf}}
We first show the equality ${\ind{\vect{P}'}{n, 1}}^{\top}\vect{1}_{m_1} = {\ind{\vect{P}'}{n, 2}}\vect{1}_{m_3}$. 
\begin{align*}
    {\ind{\vect{P}'}{n, 1}}^{\top}\vect{1}_{m_1} &= {\ind{\vect{u}}{n, 1}}^{\top}\ind{\vect{K}}{1}\diag{\frac{\vect{1}_{m_{2}}}{\ind{\vect{u}}{n+1, 2}}} \\
    &= {\ind{\vect{u}}{n, 1}}^{\top}\ind{\vect{K}}{1}\diag{\Big(\frac{\ind{\vect{K}}{3}\ind{\vect{u}}{n,3}}{(\ind{\vect{K}}{1})^{\top} \ind{\vect{u}}{n,1}}\Big)^{1/2}}\\
     &= {\ind{\vect{u}}{n, 1}}^{\top}\ind{\vect{K}}{1}\diag{\frac{\vect{1}_{m_2}}{\ind{\vect{u}}{n, 2}}}\diag{{\ind{\vect{u}}{n, 2}}}\diag{\Big(\frac{\ind{\vect{K}}{3}\ind{\vect{u}}{n,3}}{(\ind{\vect{K}}{1})^{\top} \ind{\vect{u}}{n,1}}\Big)^{1/2}}\\
     &={\vect{1}_{m_1}}^{\top} \ind{\vect{P}}{n, 1}\diag{\Bigg( \frac{\ind{\vect{P}}{n,2}\vect{1}_{m_3}}{\vect{1}_{m_1}^{\top}\ind{\vect{P}}{n, 1}}\Bigg)^{1/2}} \\
     &= \bigg( \Big(\sum_{j}\ind{P}{n, 1}_{jk}\Big)^{1/2} \Big(\sum_l \ind{P}{n, 2}_{kl}\Big)^{1/2} \bigg)_{k\in \nset{m_2}}.
\end{align*}
Similarly, we can easily show that 
\begin{align*}
    {\ind{\vect{P}'}{n, 2}}\vect{1}_{m_3} = \bigg( \Big(\sum_{j}\ind{P}{n, 1}_{jk}\Big)^{1/2} \Big(\sum_l \ind{P}{n, 2}_{kl}\Big)^{1/2} \bigg)_{k\in \nset{m_2}}. 
\end{align*}

Next, we show the equality $\ind{\vect{P}'}{n, 1}\vect{1}_{m_2} = \ind{\vect{P}}{n, 1}\bigg( \frac{\ind{\vect{P}}{n,2}\vect{1}_{m_3}}{\vect{1}_{m_1}^{\top}\ind{\vect{P}}{n, 1}}\bigg)^{1/2}$. 
\begin{align*}
    \ind{\vect{P}'}{n, 1}\vect{1}_{m_2} &= \diag{{\ind{\vect{u}}{n, 1}}^{\top}}\ind{\vect{K}}{1}\frac{\vect{1}_{m_{2}}}{\ind{\vect{u}}{n+1, 2}}\\
    &= \diag{{\ind{\vect{u}}{n, 1}}^{\top}}\ind{\vect{K}}{1}\bigg(\frac{\ind{\vect{K}}{3}\ind{\vect{u}}{n,3}}{(\ind{\vect{K}}{1})^{\top} \ind{\vect{u}}{n,1}}\bigg)^{1/2}\\
    &= \diag{{\ind{\vect{u}}{n, 1}}^{\top}}\ind{\vect{K}}{1}\diag{\frac{\vect{1}_{m_2}}{\ind{\vect{u}}{n, 2}}}\diag{{\ind{\vect{u}}{n, 2}}}\bigg(\frac{\ind{\vect{K}}{3}\ind{\vect{u}}{n,3}}{(\ind{\vect{K}}{1})^{\top} \ind{\vect{u}}{n,1}}\bigg)^{1/2}\\
    &= \ind{\vect{P}}{n, 1}\bigg( \frac{\ind{\vect{P}}{n,2}\vect{1}_{m_3}}{\vect{1}_{m_1}^{\top}\ind{\vect{P}}{n, 1}}\bigg)^{1/2}.
\end{align*}
The proof of ${\ind{\vect{P}'}{n, 2}}^{\top}\vect{1}_{m_2} = (\ind{\vect{P}}{n, 2})^{\top}\bigg( \frac{\vect{1}_{m_1}^{\top}\ind{\vect{P}}{n, 1}}{\ind{\vect{P}}{n,2}\vect{1}_{m_3}}\bigg)^{1/2}$ is similar. \qquad $\blacksquare$

\subsection{Proof of Lem.~\ref{lem:maxdistanceLagrange}}
By Lem.~\ref{lem:transportmat}, 
\begin{align*}
  &\laglange\big((\ind{\opt{\vect{u}}}{i})_{i\in \nset{3}}\big) - \laglange\big((\ind{\vect{u}}{0, i})_{i\in \nset{3}}\big) \\
  = & \epsilon \vect{a}^{\top}\log\Big(\ind{\opt{\vect{u}}}{1}\Big) + \epsilon \vect{b}^{\top}\log\Big(\ind{\opt{\vect{u}}}{3}\Big) - \epsilon \vect{a}^{\top} \log\bigg(\frac{\vect{1}_{m_1}}{\|\ind{\vect{K}}{1}\|_1}\bigg) -  \epsilon \vect{b}^{\top} \log\bigg(\frac{\vect{1}_{m_3}}{\|\ind{\vect{K}}{2}\|_1}\bigg)\\
  = & \epsilon \vect{a}^{\top}\log\Big(\ind{\opt{\vect{u}}}{1}\Big) + \epsilon \vect{b}^{\top}\log\Big(\ind{\opt{\vect{u}}}{3}\Big) + \epsilon \log\Big(\big\|\ind{\vect{K}}{1}\big\|_1\Big) + \epsilon \log\Big(\big\|\ind{\vect{K}}{2}\big\|_1\Big) 
\end{align*}

Then, for any $j\in \nset{m_1}$ and $l\in \nset{m_3}$, we prove the inequality $\log\Big(\ind{\opt{u}}{1}_j\Big) + \log\Big(\ind{\opt{u}}{3}_l\Big) \leq \log\Big(\frac{1}{K}\Big)$ in the following way. 
\begin{align*}
  &\ind{\opt{u}}{1}_j\Big(\sum_{k} \ind{K}{1}_{jk}\ind{K}{2}_{kl} \Big)\ind{\opt{u}}{2}_l = \sum_{k}\ind{\opt{u}}{1}_j \ind{K}{1}_{jk}\ind{K}{2}_{kl} \ind{\opt{u}}{2}_l = \sum_{k}\ind{\opt{u}}{1}_j \ind{K}{1}_{jk}\frac{1}{\ind{\opt{u}}{2}_k} \ind{\opt{u}}{2}_k  \ind{K}{2}_{kl} \ind{\opt{u}}{2}_l \\
  =&\sum_k \ind{\opt{P}}{1}_{jk}\ind{\opt{P}}{2}_{kl}\leq 1 
\end{align*}
Therefore, the following inequalities hold:
\begin{align*}
  \log\Big(\ind{\opt{u}}{1}_j\Big) +  \log\Big(\ind{\opt{u}}{3}_l\Big)&\leq -\log\Big(\sum_k\ind{K}{1}_{jk}\ind{K}{2}_{kl}\Big) \\
  &\leq \max_{j, l}\log\Bigg(\frac{1}{\sum_k\ind{K}{1}_{jk}\ind{K}{2}_{kl}}\Bigg) \\
  &\leq \log\Bigg(\frac{1}{\min_{j, l}\sum_k\ind{K}{1}_{jk}\ind{K}{2}_{kl}}\Bigg) 
\end{align*} 
We can thus conclude that 
\begin{align*}
  \laglange\big((\ind{\opt{\vect{u}}}{i})_{i\in \nset{3}}\big) - \laglange\big((\ind{\vect{u}}{0, i})_{i\in \nset{3}}\big) 
  &\leq \epsilon\log\Big(\frac{1}{K}\Big)  + \epsilon \log\Big(\big\|\ind{\vect{K}}{1}\big\|_1\Big) + \epsilon \log\Big(\big\|\ind{\vect{K}}{2}\big\|_1\Big)\\
  &=  \epsilon \log\Bigg(\frac{\big\|\ind{\vect{K}}{1}\big\|_1\big\|\ind{\vect{K}}{2}\big\|_1}{K}\Bigg) \quad \blacksquare
\end{align*}

\subsection{Proof of Lem.~\ref{lem:terminating}}
For any $k\in \nat$ such that 
\begin{align*}
   \delta < \Bigg\|\vect{a} - \ind{\vect{P}}{n, 1}\bigg( \frac{\ind{\vect{P}}{n,2}\vect{1}_{m_3}}{\vect{1}_{m_1}^{\top}\ind{\vect{P}}{n, 1}}\bigg)^{1/2}\Bigg\|_1 + \Bigg\|\vect{b} - (\ind{\vect{P}}{n, 2})^{\top}\bigg( \frac{\vect{1}_{m_1}^{\top}\ind{\vect{P}}{n, 1}}{\ind{\vect{P}}{n,2}\vect{1}_{m_3}}\bigg)^{1/2}\Bigg\|_1,
  \end{align*}
  we can see that 
  \begin{align*}
      \delta^2 &< \Bigg ( \Bigg\|\vect{a} - \ind{\vect{P}}{n, 1}\bigg( \frac{\ind{\vect{P}}{n,2}\vect{1}_{m_3}}{\vect{1}_{m_1}^{\top}\ind{\vect{P}}{n, 1}}\bigg)^{1/2}\Bigg\|_1 + \Bigg\|\vect{b} - (\ind{\vect{P}}{n, 2})^{\top}\bigg( \frac{\vect{1}_{m_1}^{\top}\ind{\vect{P}}{n, 1}}{\ind{\vect{P}}{n,2}\vect{1}_{m_3}}\bigg)^{1/2}\Bigg\|_1\Bigg)^2\\
      &\leq 2\Bigg(\Bigg\|\vect{a} - \ind{\vect{P}}{n, 1}\bigg( \frac{\ind{\vect{P}}{n,2}\vect{1}_{m_3}}{\vect{1}_{m_1}^{\top}\ind{\vect{P}}{n, 1}}\bigg)^{1/2}\Bigg\|^2_1 + \Bigg\|\vect{b} - (\ind{\vect{P}}{n, 2})^{\top}\bigg( \frac{\vect{1}_{m_1}^{\top}\ind{\vect{P}}{n, 1}}{\ind{\vect{P}}{n,2}\vect{1}_{m_3}}\bigg)^{1/2}\Bigg\|^2_1\Bigg)\\
      &\leq 4\Bigg( \kld[\Bigg]{\vect{a}}{\ind{\vect{P}}{n, 1}\bigg( \frac{\ind{\vect{P}}{n,2}\vect{1}_{m_3}}{\vect{1}_{m_1}^{\top}\ind{\vect{P}}{n, 1}}\bigg)^{1/2}} + \kld[\Bigg]{\vect{b}}{(\ind{\vect{P}}{n, 2})^{\top}\bigg( \frac{\vect{1}_{m_1}^{\top}\ind{\vect{P}}{n, 1}}{\ind{\vect{P}}{n,2}\vect{1}_{m_3}}\bigg)^{1/2}}\Bigg)
  \end{align*}
The second inequality comes from the equality $(x+y)^2 + (x-y)^2 = 2x^2 + 2y^2$ for any $x, y\in \real$, and the last inequality comes from the Pinsker's inequality (e.g.~\cite{slivkins2019introduction}). We note that the Pinsker's inequality $\|P-Q\|^2_1\leq 2\kld{P}{Q}$ holds even for the case that $P$ is a distribution and $Q$ is a sub-distribution. This is because by extending the support, we can create a distribution $Q'$ such that  $\|P-Q\|^2_1 \leq \|P-Q'\|^2_1 \leq 2\kld{P}{Q'} = 2\kld{P}{Q}$ by the Pinsker's inequality for the two distributions $P$ and $Q'$.

By Lem.~\ref{lem:diffLagrange}, we now know that the value of Lagrangian increases at least $\frac{\delta^2}{f}\epsilon$ for each iteration. 
By Lem.~\ref{lem:maxdistanceLagrange}, we conclude that this increase happens at most $\frac{4}{\delta^2} \log\Big(\frac{\|\ind{K}{1}\|_1\|\ind{K}{2}\|_1}{K}\Big)$ times. Thus, the iteration terminates within at most $1 + \frac{4}{\delta^2} \log\Big(\frac{\|K_1\|_1\|K_2\|_1}{K}\Big) $ times. \quad $\blacksquare$

\subsection{Proof of Thm.~\ref{thm:main_theorem}}
Since $\big\|\ind{\vect{P}}{k +0.5, 1}\big\|_1 = \big\|\ind{\vect{P}}{k +0.5, 2}\big\|_1 = 1$ and by Lem.~\ref{lem:ch_optimal_matrix}, 
we know that the pair $\ind{\vect{P}}{k +0.5, 1}$ and $\ind{\vect{P}}{k +0.5, 2}$ are an optimal solution for the distributions $\vect{a}' \defeq\ind{\vect{P}}{k +0.5, 1}\vect{1}_{m_2} $ and $\vect{b}' \defeq{\ind{\vect{P}}{k +0.5, 2}}^{\top}\vect{1}_{m_2}$.
By their construction, we know that 
\begin{align*}
  \vect{a}' &= \frac{1}{P}\ind{\vect{P}}{n, 1}\Bigg( \frac{\ind{\vect{P}}{n,2}\vect{1}_{m_3}}{\vect{1}_{m_1}^{\top}\ind{\vect{P}}{n, 1}}\Bigg)^{1/2}, & \vect{b}' &= \frac{1}{P}(\ind{\vect{P}}{n, 2})^{\top}\Bigg( \frac{\vect{1}_{m_1}^{\top}\ind{\vect{P}}{n, 1}}{\ind{\vect{P}}{n,2}\vect{1}_{m_3}}\Bigg)^{1/2},
\end{align*}
where $P \defeq \big\| \diag{\ind{\vect{u}}{n+1, 2}}\ind{\vect{K}}{2}\diag{\ind{\vect{u}}{n, 3}}\big\|_1$  (that is defined in Def.~\ref{def:halfPlan}). 

Suppose that a pair of $\ind{\breve{\vect{P}}}{1}$ and $\ind{\breve{\vect{P}}}{2}$ is an optimal solution of the original problem (Def.~\ref{def:seq_composed_optimal_transport}). 
By Lem.~\ref{lem:getFeasibleSolution}, we obtain that the tranportation plans $\widetilde{\ind{\breve{\vect{P}}}{1}}$ and $\widetilde{\ind{\breve{\vect{P}}}{2}}$ for $\vect{a}'$ and $\vect{b}'$ and the error is bounded by 
\begin{align*}
  \Big\|\widetilde{\ind{\breve{\vect{P}}}{1}}- \ind{\breve{\vect{P}}}{1}\Big\|_1 \leq 2\big\| \vect{a} - \vect{a}' \big\|_1, && \Big\|\widetilde{\ind{\breve{\vect{P}}}{2}}- \ind{\breve{\vect{P}}}{2}\Big\|_1 \leq 2\big\| \vect{b} - \vect{b}' \big\|_1.
\end{align*}
Since  $\ind{\vect{P}}{k +0.5, 1}$ and $\ind{\vect{P}}{k +0.5, 2}$ are an optimal solution for the distributions $\vect{a}'$ and $\vect{b}'$, we obtain the following inequality
\begin{align*}
  \sum_{i\in \nset{2}}\Big\langle\ind{\vect{C}}{i}, \ind{\vect{P}}{k +0.5, i}\Big\rangle - \epsilon\ent{\ind{\vect{P}}{k +0.5, i}} &\leq \sum_{i\in \nset{2}}\Big\langle\ind{\vect{C}}{i}, \widetilde{\ind{\breve{\vect{P}}}{i}}\Big\rangle - \epsilon\ent{\widetilde{\ind{\breve{\vect{P}}}{i}}},\\
  \sum_{i\in \nset{2}}\Big\langle\ind{\vect{C}}{i}, \ind{\vect{P}}{k +0.5, i}\Big\rangle - \Big\langle\ind{\vect{C}}{i}, \widetilde{\ind{\breve{\vect{P}}}{i}}\Big\rangle  &\leq \sum_{i\in \nset{2}} \epsilon\ent{\ind{\vect{P}}{k +0.5, i}} - \epsilon\ent{\widetilde{\ind{\breve{\vect{P}}}{i}}}\\
  &\leq \epsilon\big(\log(m_1 m_2) + \log(m_2m_3) \big).\\
\end{align*}
The later inequalities follows from the fact $1\leq \ent{\vect{P}} \leq 1 + \log(mn)$ for any $\vect{P}\in  \Rnnegmat{m}{n}$ s.t. $\|\vect{P}\|_1 = 1$. 

By Lem.~\ref{lem:getFeasibleSolution}, we know the error is bounded by 
\begin{align*}
  \bigg\|\widetilde{\ind{\vect{P}}{k +0.5, 1}} - \ind{\vect{P}}{k +0.5, 1}\bigg\|_1 \leq 2\big\|\vect{a} - \vect{a}'\big\|, && \bigg\|\widetilde{\ind{\vect{P}}{k +0.5, 2}} - \ind{\vect{P}}{k +0.5, 2}\bigg\|_1 \leq 2\big\|\vect{b} - \vect{b}'\big\|.
\end{align*}

Then, we evaluate the error by 
\begin{align*}
  &\sum_{i\in \nset{2}} \Big\langle \ind{\vect{C}}{i},\widetilde{\ind{\vect{P}}{k +0.5, i}}\Big\rangle - \sum_{i\in \nset{2}} \Big\langle \ind{\vect{C}}{i},\ind{\breve{\vect{P}}}{i}\Big\rangle \\
  \leq& \sum_{i\in \nset{2}} \Big\langle \ind{\vect{C}}{i},\widetilde{\ind{\vect{P}}{k +0.5, i}}\Big\rangle - \Big\langle \ind{\vect{C}}{i},\ind{\vect{P}}{k +0.5, i}\Big\rangle  + \Big\langle \ind{\vect{C}}{i},\ind{\vect{P}}{k +0.5, i}\Big\rangle  \\
  &- \Big\langle \ind{\vect{C}}{i},\widetilde{\ind{\breve{\vect{P}}}{i}}\Big\rangle + \Big\langle \ind{\vect{C}}{i},\widetilde{\ind{\breve{\vect{P}}}{i}}\Big\rangle - \Big\langle \ind{\vect{C}}{i},\ind{\breve{\vect{P}}}{i}\Big\rangle \\
  \leq& 4\big(\|\vect{a}- \vect{a}'\|_1 + \|\vect{b}- \vect{b}'\|_1\big)\max\big(\|\vect{C}_1\|_{\infty},\|\vect{C}_2\|_{\infty} \big) + \epsilon\big(\log(m_1 m_2) + \log(m_2m_3) \big). 
\end{align*}
In the last inequality, we use the H\"{o}lder's inequality. 

We evaluate the errors $\big\|\vect{a}- \vect{a}'\big\|_1 $ and $\big\|\vect{b}- \vect{b}'\big\|_1 $. By the triangle inequality, 
\begin{align*}
  \big\|\vect{a}- \vect{a}'\big\|_1  &\leq \Bigg\|\vect{a}- \ind{\vect{P}}{n, 1}\bigg( \frac{\ind{\vect{P}}{n,2}\vect{1}_{m_3}}{\vect{1}_{m_1}^{\top}\ind{\vect{P}}{n, 1}}\bigg)^{1/2}\Bigg\|_1 + \Bigg\|\ind{\vect{P}}{n, 1}\bigg( \frac{\ind{\vect{P}}{n,2}\vect{1}_{m_3}}{\vect{1}_{m_1}^{\top}\ind{\vect{P}}{n, 1}}\bigg)^{1/2} - \vect{a}'\Bigg\|_1 \\
  & = \Bigg\|\vect{a}- \ind{\vect{P}}{n, 1}\bigg( \frac{\ind{\vect{P}}{n,2}\vect{1}_{m_3}}{\vect{1}_{m_1}^{\top}\ind{\vect{P}}{n, 1}}\bigg)^{1/2}\Bigg\|_1 + 1-P\\
  &\leq 2\Bigg\|\vect{a}- \ind{\vect{P}}{n, 1}\bigg( \frac{\ind{\vect{P}}{n,2}\vect{1}_{m_3}}{\vect{1}_{m_1}^{\top}\ind{\vect{P}}{n, 1}}\bigg)^{1/2}\Bigg\|_1
\end{align*}
The last inequality holds because 
\begin{align*}
  1 - P =  \sum_{i} a_i -  \Bigg(\ind{\vect{P}}{n, 1}\bigg( \frac{\ind{\vect{P}}{n,2}\vect{1}_{m_3}}{\vect{1}_{m_1}^{\top}\ind{\vect{P}}{n, 1}}\bigg)^{1/2}\Bigg)_i \leq \Bigg\|\vect{a}- \ind{\vect{P}}{n, 1}\bigg( \frac{\ind{\vect{P}}{n,2}\vect{1}_{m_3}}{\vect{1}_{m_1}^{\top}\ind{\vect{P}}{n, 1}}\bigg)^{1/2}\Bigg\|_1
\end{align*}
Similarly, we get 
\begin{align*}
  \big\|\vect{b}- \vect{b}'\big\|_1  &\leq 2\Bigg\|\vect{b}- (\ind{\vect{P}}{n, 2})^{\top}\bigg( \frac{\vect{1}_{m_1}^{\top}\ind{\vect{P}}{n, 1}}{\ind{\vect{P}}{n,2}\vect{1}_{m_3}}\bigg)^{1/2}\Bigg\|_1
\end{align*}
Therefore, we obtain that 
\begin{align*}
  &\sum_{i\in \nset{2}} \Big\langle \ind{\vect{C}}{i},\widetilde{\ind{\vect{P}}{k +0.5, i}}\Big\rangle - \sum_{i\in \nset{2}} \Big\langle \ind{\vect{C}}{i},\ind{\breve{\vect{P}}}{i}\Big\rangle \\
  \leq & 8\Bigg(\Bigg\|\vect{a}- \ind{\vect{P}}{n, 1}\Big( \frac{\ind{\vect{P}}{n,2}\vect{1}_{m_3}}{\vect{1}_{m_1}^{\top}\ind{\vect{P}}{n, 1}}\Big)^{1/2}\Bigg\|_1 +  \Bigg\|\vect{b}- (\ind{\vect{P}}{n, 2})^{\top}\Big( \frac{\vect{1}_{m_1}^{\top}\ind{\vect{P}}{n, 1}}{\ind{\vect{P}}{n,2}\vect{1}_{m_3}}\Big)^{1/2}\Bigg\|_1\Bigg)\max\big(\|\vect{C}_1\|_{\infty},\|\vect{C}_2\|_{\infty} \big)\\
  & + \epsilon\big(\log(m_1 (m_2)^2m_3) \big)\\
  \leq & \delta 
\end{align*}
We conclude the proof by evaluating the worst-case time complexity. 
By Lem.~\ref{lem:terminating}, we know the iteration terminates in at most 
\begin{align*}
  &1 + \frac{4\cdot 16^2 \cdot \big(\max(\|\ind{C}{1}\|_{\infty}, \|\ind{C}{2}\|_{\infty})\big)^2}{\delta^2}\log\Bigg(\frac{\|\ind{\vect{K}}{1}\|_1\|\ind{\vect{K}}{2}\|_1}{K}\Bigg)
\end{align*}
and we get 
\begin{align*}
  \log\Big(\big\|\ind{\vect{K}}{1}\big\|_1\big\|\ind{\vect{K}}{2}\big\|_1\Big) &\leq \log\big(m_1 (m_2)^2 m_3\big),
\end{align*}
and 
\begin{align*}
  \log(K) &= \log\Big(\min_{i, j} \sum_k\ind{K}{1}_{ik}\ind{K}{2}_{kj}\Big) = \min_{i, j} \log\Big(\sum_k\ind{K}{1}_{ik}\ind{K}{2}_{kj}\Big) \\
  &= \min_{i, j} \log\Bigg(\sum_k \exp\bigg(\frac{-\ind{C}{1}_{ik} - \ind{C}{2}_{kj} }{\epsilon}\bigg)\Bigg) \\
  &\geq \min_{i, j} \max_k \frac{-\ind{C}{1}_{ik} - \ind{C}{2}_{kj} }{\epsilon}\\
  & = - \frac{1}{\epsilon}\Big(\min_{i, j, k} \ind{C}{1}_{ik} + \ind{C}{2}_{kj}\Big)\\
  &\geq - \frac{1}{\epsilon}\|\vect{C}\|_{\infty}\\
  & = - \frac{2\log\big(m_1(m_2)^2 m_3\big)}{\delta}\|\vect{C}\|_{\infty}\\
\end{align*}
For each iteration, the update of vectors requires $O\big(\max(m_1m_2, m_2m_3)\big)$ and the construcions of transportation plans are also done in $O\big(\max(m_1m_2, m_2m_3)\big)$. 
Together with the above two inequalities, we obtain the desired result. $\qquad \blacksquare$

\subsection{Proof of Lem.~\ref{lem:diff1}}
Let $\vect{x}\defeq \bigg( \frac{\vect{1}_{m_1}^{\top}\ind{\vect{P}}{n, 1}}{\ind{\vect{P}}{n,2}\vect{1}_{m_3}}\bigg)^{1/2}$.
We prove the statement in the following: 
\begin{align*}
  &\Bigg\|\vect{a} - \ind{\vect{P}}{n, 1}\bigg( \frac{\ind{\vect{P}}{n,2}\vect{1}_{m_3}}{\vect{1}_{m_1}^{\top}\ind{\vect{P}}{n, 1}}\bigg)^{1/2}\Bigg\|_1 + \Bigg\|\vect{b} - (\ind{\vect{P}}{n, 2})^{\top}\bigg( \frac{\vect{1}_{m_1}^{\top}\ind{\vect{P}}{n, 1}}{\ind{\vect{P}}{n,2}\vect{1}_{m_3}}\bigg)^{1/2}\Bigg\|_1\\
=&\sum_j\bigg|\sum_k \ind{P}{n, 1}_{jk} - \sum_k \ind{P}{n, 1}_{jk}\frac{1}{x_k}\bigg| + \sum_l\bigg|\sum_k \ind{P}{n, 2}_{kl} - \sum_k \ind{P}{n, 2}_{kl}x_k\bigg|\\
\leq&\sum_k\Bigg(\sum_j\ind{P}{n, 1}_{jk} \bigg|1 - \frac{1}{x_k}\bigg|  + \sum_{l} \ind{P}{n, 2}_{kl} \bigg|1 - x_k\bigg|  \Bigg)\\
\leq&\Big\|\ind{\vect{P}}{n, 2}\vect{1}_{m_3}- (\ind{\vect{P}}{n, 1})^{\top}\vect{1}_{m_1}\Big\|_1 + \sum_k \bigg| \sum_j  \ind{P}{n, 1}_{jk} \frac{1}{x_k} - \sum_l  \ind{P}{n, 2}_{kl} x_k \bigg|\\
= &\Big\|\ind{\vect{P}}{n, 2}\vect{1}_{m_3}- (\ind{\vect{P}}{n, 1})^{\top}\vect{1}_{m_1}\Big\|_1 
\end{align*}
The last equality holds because 
\begin{align*}
  &\sum_k \Big| \sum_j  \ind{P}{n, 1}_{jk} \frac{1}{x_k} - \sum_l  \ind{P}{n, 2}_{kl} x_k \Big|\\
=& \sum_k\Big| \big(\sum_j \ind{P}{n, 1}_{jk} \big)^{1/2}\cdot \big(\sum_l \ind{P}{n, 2}_{kl} \big)^{1/2} - \big(\sum_j \ind{P}{n, 1}_{jk} \big)^{1/2}\cdot \big(\sum_l \ind{P}{n, 2}_{kl} \big)^{1/2}\Big| = 0 \quad \blacksquare
\end{align*}

\subsection{Proof of Lem.~\ref{lem:diff2}}

Let $\vect{x}$ be $ x_k \defeq \Big(\big(\sum_j \ind{P}{n,1}_{jk}\big)\big(\sum_l \ind{P}{n,2}_{kl}\big)\Big)^{1/2}$ for each $k\in \nset{m_2}$. 
Then, we obtain the following inequality
\begin{align*}
  &\Big\|\ind{\vect{P}}{n+1, 2}\vect{1}_{m_3}- (\ind{\vect{P}}{n+1, 1})^{\top}\vect{1}_{m_1}\Big\|_1 \\
  \leq & \Big\|\ind{\vect{P}}{n+1, 2}\vect{1}_{m_3}- \vect{x}\Big\|_1 + \Big\| (\ind{\vect{P}}{n+1, 1})^{\top}\vect{1}_{m_1} - \vect{x}\Big\|_1.
\end{align*}
Let $\vect{y}\defeq \bigg( \frac{\vect{1}_{m_1}^{\top}\ind{\vect{P}}{n, 1}}{\ind{\vect{P}}{n,2}\vect{1}_{m_3}}\bigg)^{1/2}$.
We can evaluate $\Big\|\ind{\vect{P}}{n+1, 2}\vect{1}_{m_3}- \vect{x}\Big\|_1$ as follows: 
\begin{align*}
  &\Big\|\ind{\vect{P}}{n+1, 2}\vect{1}_{m_3}- \vect{x}\Big\|_1\\
  = &\sum_k\Big| \big(\sum_l \ind{\vect{P}}{n+1, 2}_{kl}\big) -  \big(\sum_j \ind{P}{n,1}_{jk}\big)^{1/2}\big(\sum_l \ind{P}{n,2}_{kl}\big)^{1/2}\Big|
\end{align*}
Here, we note that 
\begin{align*}
  &\sum_l \ind{\vect{P}}{n+1, 2}_{kl} = \sum_l \ind{u}{n+1, 2}_k\ind{K}{2}_{kl}\ind{u}{n+1, 3}_l\\
   = &\sum_l\bigg(\Big(\frac{(\ind{\vect{K}}{1})^{\top} \ind{\vect{u}}{n,1}}{\ind{\vect{K}}{2}\ind{\vect{u}}{n,3}}\Big)^{1/2}\bigg)_k \ind{K}{2}_{kl}\Bigg(\frac{\vect{b}}{(\ind{\vect{K}}{2})^{\top}\Big(\frac{(\ind{\vect{K}}{1})^{\top} \ind{\vect{u}}{n,1}}{\ind{\vect{K}}{2}\ind{\vect{u}}{n,3}}\Big)^{1/2}}\Bigg)_l\\
   = & \sum_l\bigg(\Big(\frac{(\ind{\vect{P}}{n, 1})^{\top} \vect{1}_{m_1}}{\ind{\vect{P}}{n, 2}\vect{1}_{m_3}}\Big)^{1/2}\bigg)_k \ind{u}{n,2}_k \ind{K}{2}_{kl}\Bigg(\frac{\vect{b}}{(\ind{\vect{K}}{2})^{\top}\Big(\frac{(\ind{\vect{K}}{1})^{\top} \ind{\vect{u}}{n,1}}{\ind{\vect{K}}{2}\ind{\vect{u}}{n,3}}\Big)^{1/2}}\Bigg)_l\\
   = & \sum_l\bigg(\Big(\frac{(\ind{\vect{P}}{n, 1})^{\top} \vect{1}_{m_1}}{\ind{\vect{P}}{n, 2}\vect{1}_{m_3}}\Big)^{1/2}\bigg)_k \ind{u}{n,2}_k \ind{K}{2}_{kl}\ind{u}{n,3}_l \Bigg(\frac{\vect{b}}{\diag{\ind{u}{n,3}}(\ind{\vect{K}}{2})^{\top}\Big(\frac{(\ind{\vect{K}}{1})^{\top} \ind{\vect{u}}{n,1}}{\ind{\vect{K}}{2}\ind{\vect{u}}{n,3}}\Big)^{1/2}}\Bigg)_l\\
   = & \sum_l y_k \ind{u}{n,2}_k \ind{K}{2}_{kl}\ind{u}{n,3}_l \Bigg(\frac{\vect{b}}{(\ind{\vect{P}}{n, 2})^{\top}\vect{y}}\Bigg)_l
\end{align*}
Therefore, 
\begin{align*}
  &\Big\|\ind{\vect{P}}{n+1, 2}\vect{1}_{m_3}- \vect{x}\Big\|_1\\
  =&\sum_k y_k\Bigg|\bigg(\sum_l \ind{P}{n, 2}_{kl} \Big(\frac{\vect{b}}{(\ind{\vect{P}}{n, 2})^{\top}\vect{y}}\Big)_l \bigg) - \Big(\sum_l  \ind{P}{n, 2}_{kl}\Big)  \Bigg|\\
  \leq &\sum_k y_k\sum_l \Bigg|\ind{P}{n, 2}_{kl} \bigg(\frac{\vect{b}}{(\ind{\vect{P}}{n, 2})^{\top}\vect{y}}\bigg)_l  - \ind{P}{n, 2}_{kl}\Bigg|\\
  = &\sum_k y_k\sum_l \ind{P}{n, 2}_{kl} \Bigg|\bigg(\frac{\vect{b}}{(\ind{\vect{P}}{n, 2})^{\top}\vect{y}}\bigg)_l  - 1\Bigg|\\
  = & \sum_l \sum_k  \ind{P}{n, 2}_{kl}y_k\Bigg|\frac{b_l}{ \sum_{k'}  \ind{P}{n, 2}_{k'l}y_{k'}}  - 1\Bigg|
\end{align*}
For each $l$, the sign of $\frac{b_l}{ \sum_{k'}  \ind{P}{n, 2}_{k'l}y_{k'}}  - 1$ is invariant with respect to $k$, therefore 
\begin{align*}
  &\Big\|\ind{\vect{P}}{n+1, 2}\vect{1}_{m_3}- \vect{x}\Big\|_1\\
  \leq & \sum_l \Bigg|\sum_k  \big(\ind{P}{n, 2}_{kl}y_k\big)\bigg(\frac{b_l}{ \sum_{k'}  \ind{P}{n, 2}_{k'l}y_{k'}}  - 1\bigg)\Bigg|\\
  = & \sum_l \Bigg|b_l - \sum_k  \big(\ind{P}{n, 2}_{kl}y_k\big)\Bigg|\\
  = & \Big\|\vect{b} - (\ind{\vect{P}}{n, 2})^{\top}\bigg( \frac{\vect{1}_{m_1}^{\top}\ind{\vect{P}}{n, 1}}{\ind{\vect{P}}{n,2}\vect{1}_{m_3}}\bigg)^{1/2}\Big\|_1
\end{align*}
Let $\vect{z}\defeq\frac{\vect{1}_{m_2}}{\vect{y}}  =  \bigg( \frac{\ind{\vect{P}}{n,2}\vect{1}_{m_3}}{\vect{1}_{m_1}^{\top}\ind{\vect{P}}{n, 1}}\bigg)^{1/2}$.
We evaluate $\Big\| (\ind{\vect{P}}{n+1, 1})^{\top}\vect{1}_{m_1} - \vect{x}\Big\|_1$ as follows: 
\begin{align*}
  &\Big\| (\ind{\vect{P}}{n+1, 1})^{\top}\vect{1}_{m_1} - \vect{x}\Big\|_1\\
  = & \sum_k\bigg| \big(\sum_j \ind{\vect{P}}{n+1, 1}_{jk}\big) -  \big(\sum_j \ind{P}{n,1}_{jk}\big)^{1/2}\big(\sum_l \ind{P}{n,2}_{kl}\big)^{1/2}\bigg|
\end{align*}
Here we note that 
\begin{align*}
  &\sum_j \ind{\vect{P}}{n+1, 1}_{jk} =  \sum_j \ind{u}{n+1, 1}_j\ind{K}{1}_{jk}\frac{1}{\ind{u}{n+1, 2}_k}\\
  = &\sum_j \Bigg(\frac{\vect{a}}{\ind{\vect{K}}{1}\Big(\frac{\ind{\vect{K}}{2}\ind{\vect{u}}{n,3}}{(\ind{\vect{K}}{1})^{\top} \ind{\vect{u}}{n,1}}\Big)^{1/2}}\Bigg)_j\ind{K}{1}_{jk} \Bigg(\bigg(\frac{\ind{\vect{K}}{2}\ind{\vect{u}}{n,3}}{(\ind{\vect{K}}{1})^{\top} \ind{\vect{u}}{n,1}}\bigg)^{1/2}\Bigg)_k\\
  = &\sum_j \Bigg(\frac{\vect{a}}{\ind{\vect{K}}{1}\Big(\frac{\ind{\vect{K}}{2}\ind{\vect{u}}{n,3}}{(\ind{\vect{K}}{1})^{\top} \ind{\vect{u}}{n,1}}\Big)^{1/2}}\Bigg)_j\ind{K}{1}_{jk} \frac{1}{\ind{u}{n, 2}_k} z_k\\
  = &\sum_j \Bigg(\frac{\vect{a}}{\diag{\ind{u}{n, 1}}\ind{\vect{K}}{1}\Big(\frac{\ind{\vect{K}}{2}\ind{\vect{u}}{n,3}}{(\ind{\vect{K}}{1})^{\top} \ind{\vect{u}}{n,1}}\Big)^{1/2}}\Bigg)_j \ind{u}{n, 1}_j \ind{K}{1}_{jk} \frac{1}{\ind{u}{n, 2}_k} z_k\\
  = &\sum_j \bigg(\frac{\vect{a}}{\ind{\vect{P}}{n, 1}\vect{z}}\bigg)_j \ind{u}{n, 1}_j \ind{K}{1}_{jk} \frac{1}{\ind{u}{n, 2}_k} z_k\\
\end{align*}
Therefore, 
\begin{align*}
  &\Big\| (\ind{\vect{P}}{n+1, 1})^{\top}\vect{1}_{m_1} - \vect{x}\Big\|_1\\
  = &\sum_k\bigg| \sum_j \Big(\frac{\vect{a}}{\ind{\vect{P}}{n, 1}\vect{z}}\Big)_j \ind{P}{n, 1}_{jk} z_k - x_k \bigg|\\
  = & \sum_k z_k \bigg| \sum_j \Big(\frac{\vect{a}}{\ind{\vect{P}}{n, 1}\vect{z}}\Big)_j \ind{P}{n, 1}_{jk} - \Big(\sum_j  \ind{P}{n, 1}_{jk} \Big) \bigg|\\
  \leq & \sum_k z_k \sum_j \bigg| \Big(\frac{\vect{a}}{\ind{\vect{P}}{n, 1}\vect{z}}\Big)_j \ind{P}{n, 1}_{jk} - \ind{P}{n, 1}_{jk}  \bigg|\\
  = & \sum_k z_k \sum_j \ind{P}{n, 1}_{jk}  \bigg| \Big(\frac{\vect{a}}{\ind{\vect{P}}{n, 1}\vect{z}}\Big)_j - 1 \bigg|\\
  = & \sum_j \sum_k z_k \ind{P}{n, 1}_{jk}  \bigg| \Big(\frac{\vect{a}}{\ind{\vect{P}}{n, 1}\vect{z}}\Big)_j - 1 \bigg|\\
  = & \sum_j \Big| a_j - \sum_k z_k \ind{P}{n, 1}_{jk}\Big|\\
  = & \Bigg\|\vect{a} - \ind{\vect{P}}{n, 1}\bigg( \frac{\ind{\vect{P}}{n,2}\vect{1}_{m_3}}{\vect{1}_{m_1}^{\top}\ind{\vect{P}}{n, 1}}\bigg)^{1/2}\Bigg\|_1 \qquad \blacksquare
\end{align*}

\end{document}